%% file: main.tex
\documentclass{sig-alternate-05-2015}
\usepackage{etoolbox}
\makeatletter
\patchcmd{\maketitle}{\@copyrightspace}{}{}{}
\makeatother

\usepackage{amsmath}
\allowdisplaybreaks

\usepackage{color}
\usepackage{enumitem}
\usepackage{multirow} 
\allowbreak

\graphicspath{{figures/}}
\DeclareGraphicsExtensions{.pdf,.eps}

\usepackage[font=bf]{caption}
\usepackage[labelformat=simple]{subcaption}

\usepackage{booktabs}
\usepackage{mathtools}

\newtheorem{theorem}{Theorem}
\newtheorem{conjecture}{Conjecture}
\newtheorem{lemma}{Lemma}
\newtheorem{corollary}{Corollary}
\newtheorem{definition}{Definition}
\newdef{remark}{Remark}
\newcommand{\myfill}{\hspace*{\fill}}

\usepackage{prettyref}
\newrefformat{sec}{Section~\ref{#1}}
\newrefformat{subsec}{Section~\ref{#1}}
\newrefformat{subsubsec}{Section~\ref{#1}}
\newrefformat{thm}{Theorem~\ref{#1}}
\newrefformat{THM}{THEOREM~\ref{#1}}
\newrefformat{lemma}{Lemma~\ref{#1}}
\newrefformat{LEMMA}{LEMMA~\ref{#1}}
\newrefformat{corollary}{Corollary~\ref{#1}}
\newrefformat{fig}{Figure~\ref{#1}}
\newrefformat{tab}{Table~\ref{#1}}
\newrefformat{eq}{\eqref{#1}}
\newrefformat{app}{Appendix~\ref{#1}}

\newcommand{\X}{\mathbf{X}}
\newcommand{\x}{\mathbf{x}}
\newcommand{\Y}{\mathbf{Y}}
\newcommand{\y}{\mathbf{y}}
\newcommand{\Z}{\mathbf{Z}}
\renewcommand{\P}{\mathbb{P}}
\newcommand{\E}{\mathbb{E}}
\newcommand{\N}{\mathbb{N}}
\newcommand{\R}{\mathbb{R}}
\newcommand{\ind}[1]{\textup{\textbf{1}}\left\{#1\right\}}

\renewcommand{\succeq}{\geq_{\textup{st}}}
\renewcommand{\preceq}{\leq_{\textup{st}}}
\newcommand{\figwidth}{1\columnwidth}

\begin{document}

% Copyright
\CopyrightYear{2016} 
\setcopyright{acmlicensed}
\conferenceinfo{SIGMETRICS '16,}{June 14 - 18, 2016, Antibes Juan-Les-Pins, France}
\isbn{978-1-4503-4266-7/16/06}\acmPrice{\$15.00}
\doi{http://dx.doi.org/10.1145/2896377.2901475}
%Authors, replace the red X's with your assigned DOI string.

\clubpenalty=10000 
\widowpenalty = 10000

%\CopyrightYear{2007} % Allows default copyright year (20XX) to be over-ridden - IF NEED BE.
%\crdata{0-12345-67-8/90/01}  % Allows default copyright data (0-89791-88-6/97/05) to be over-ridden - IF NEED BE.
% --- End of Author Metadata ---

\title{On the Duration and Intensity of Competitions in \\Nonlinear P\'olya Urn Processes with Fitness}
%\subtitle{UMass Computer Science Technical Report UM-CS-2016-001}

%\subtitle{[Extended Abstract]
%\titlenote{A full version of this paper is available as
%\textit{Author's Guide to Preparing ACM SIG Proceedings Using
%\LaTeX$2_\epsilon$\ and BibTeX} at

\numberofauthors{2} %  in this sample file, there are a *total*
\author{
% 1st. author
\alignauthor
Bo Jiang\\
      \affaddr{College of Information and Computer Sciences}\\
       \affaddr{University of Massachusetts}\\
       \affaddr{Amherst MA, USA}\\
%	 \affaddr{UMass Amherst}\\
       \email{bjiang@cs.umass.edu}
% 2nd. author
\alignauthor
Daniel R.~Figueiredo\\
       \affaddr{COPPE/PESC}\\
       \affaddr{Federal University of Rio de Janeiro (UFRJ)}\\
       \affaddr{Rio de Janeiro, Brazil}\\
%    \affaddr{UFRJ, Brazil}
       \email{daniel@land.ufrj.br}
% 3rd. author
\and
\alignauthor Bruno Ribeiro\\
      \affaddr{Department of Computer Science}\\
       \affaddr{Purdue University}\\
       \affaddr{West Lafayette IN, USA}\\
       \email{ribeiro@cs.purdue.edu}
%\and  % use '\and' if you need 'another row' of author names
% 4th. author
\alignauthor Don Towsley\\
      \affaddr{College of Information and Computer Sciences}\\
       \affaddr{University of Massachusetts}\\
       \affaddr{Amherst MA, USA}\\
%\affaddr{UMass Amherst}
       \email{towsley@cs.umass.edu}
}
% There's nothing stopping you putting the seventh, eighth, etc.
% author on the opening page (as the 'third row') but we ask,
% for aesthetic reasons that you place these 'additional authors'
% in the \additional authors block, viz.
%\additionalauthors{Additional authors: John Smith (The Th{\o}rv{\"a}ld Group,
%email: {\texttt{jsmith@affiliation.org}}) and Julius P.~Kumquat
%(The Kumquat Consortium, email: {\texttt{jpkumquat@consortium.net}}).}
\date{30 July 1999}
% Just remember to make sure that the TOTAL number of authors
% is the number that will appear on the first page PLUS the
% number that will appear in the \additionalauthors section.

\maketitle
\input{TEX/abstract}

%\begin{CCSXML}
%<ccs2012>
%<concept>
%<concept_id>10002950.10003648.10003700</concept_id>
%<concept_desc>Mathematics of computing~Stochastic processes</concept_desc>
%<concept_significance>500</concept_significance>
%</concept>
%<concept>
%<concept_id>10002950.10003648.10003703</concept_id>
%<concept_desc>Mathematics of computing~Distribution functions</concept_desc>
%<concept_significance>500</concept_significance>
%</concept>
%</ccs2012>
%\end{CCSXML}
%
%\ccsdesc[500]{Mathematics of computing~Stochastic processes}
%\ccsdesc[500]{Mathematics of computing~Distribution functions}


%
%  Use this command to print the description
%
\printccsdesc

% We no longer use \terms command


\input{TEX/intro}

\input{TEX/model}

\input{TEX/related}

\input{TEX/results}

\input{TEX/social_tagging}

\input{TEX/ordering}

\input{TEX/duration}

\input{TEX/intensity}

\input{TEX/conclusion}

%\end{document}  % This is where a 'short' article might terminate

%ACKNOWLEDGMENTS are optional
\input{TEX/ack}

\bibliographystyle{abbrv}
\bibliography{refs,auxrefs}  % sigproc.bib is the name of the Bibliography in this case
% ACM needs 'a single self-contained file'!
%\balancecolumns
\appendix
%Appendix A
\input{TEX/appendix}
%\balancecolumns % GM June 2007
% That's all folks!
\end{document}

%% file: TEX/abstract.tex
%%!TEX PS-program = pdflatexmk
%%!TEX root = ../main.tex

\begin{abstract}
Cumulative advantage (CA) refers to the notion that accumulated resources foster the accumulation of further resources in competitions, a phenomenon that has been empirically observed in various contexts. The oldest and arguably simplest mathematical model that embodies this general principle is the P{\'o}lya urn process, which finds applications in a myriad of problems. The original model captures the dynamics of competitions between two equally fit agents under linear CA effects, %, where fitness refers to the inherent ability of agents to accumulate resources,  while nonlinear reinforcement means that the chance of accumulating further resources is sub/super-linear in the amount of resources already accumulated. 
which can be readily generalized to incorporate different fitnesses and nonlinear CA effects. We study two statistics of competitions under the generalized model, namely duration (i.e., time of the last tie) and intensity (i.e., number of ties). We give rigorous mathematical characterizations of the tail distributions of both duration and intensity under the various regimes for fitness and nonlinearity, which reveal very interesting behaviors. For example, fitness superiority induces much shorter competitions in the sublinear regime while much longer competitions in the superlinear regime. Our findings can shed light on the application of P{\'o}lya urn processes in  more general contexts where fitness and nonlinearity may be present.
\end{abstract}

%% file: TEX/intro.tex
%%!TEX PS-program = pdflatexmk
%%!TEX root = ../main.tex

\section{Introduction}\label{sec:intro}

Cumulative advantage (CA) is a ubiquitous phenomenon observed in various systems where agents compete for resources. CA alludes to the capacity that accumulated resources have to foster accumulation of more resources, a principle that appears in the literature under various names such as cumulative advantage~\cite{SollaPrice76}, preferential attachment~\cite{Barab99}, ``the rich get richer'', Matthew effect~\cite{diprete2006cumulative,Merton68}, path-dependent increasing returns~\cite{Arthur:94}, and processes with feedback~\cite{Drinea:2002,oliveira2009onset}. 

The oldest and arguably simplest model that embodies CA is the P{\'o}lya urn process, which has been widely studied and applied~\cite{Polya23,Mahm08,Peman07}. In particular, one can find applications of P{\'o}lya urn model in problems that arise in most areas of science, including biology, physics, economics, and of course, computer science, with a recent example described in Section~\ref{sec:social_tagging}. In its simplest form, a P{\'o}lya urn has balls with two colors. At each round a ball is chosen uniformly at random from the urn and returned to the urn with another ball of the same color, increasing the number of balls in the urn by one. Note that drawing balls of a given color increases the chance of drawing more balls of the same color, thus embodying the CA phenomenon. 

Beyond CA, an observed and recognized characteristics in competitions is  {\em fitness}, which refers to the inherent ability of an agent to accumulate resources that does not depend on the amount of resources already accumulated. A second and more recent consideration, which has also been observed in some contexts, is that the feedback induced by accumulated resources may not be linear as in the simple P\'olya urn model. In particular, the propensity to accumulate further resources can be nonlinear in the amount of resources already accumulated. These two generalizations can be easily accommodated in the P{\'o}lya urn model by assigning a fixed fitness to each color and by selecting balls not uniformly at random from the urn. Such a model is the object under consideration in this paper (formal definition in Section~\ref{sec:model}).

Two fundamental characteristics of competitions are their duration and intensity \cite{jiang2015competition}. Duration can be measured as the time required for an agent to take the lead forever, while intensity as the number of times agents tie for the leadership.  
These two metrics have recently been studied for linear P\'olya urn processes with fitness in \cite{jiang2015competition}. The question that we ask here is: What is the impact of introducing nonlinearity in the CA feedback of a P\'olya urn process? We address this question by providing a rigorous theoretical understanding of the implications of fitness and nonlinear CA on duration and intensity, along with numerical simulations to illustrate and support the findings. A summary of our main results is given in \prettyref{sec:results}.

The rest of this paper is organized as follows. \prettyref{sec:model} formally introduces the nonlinear P{\'o}lya urn process with fitness, discusses some related work, and briefly presents a recent application in computer science. \prettyref{sec:order} presents some stochastic ordering results for the metrics investigated. Sections~\ref{sec:duration} and \ref{sec:intensity} present the main results on the distributions of duration and intensity, respectively. \prettyref{sec:conclusion} concludes the paper with further discussions.

%% file: TEX/model.tex
%%!TEX PS-program = pdflatexmk
%%!TEX root = ../main.tex

\section{Nonlinear P\'olya Urn Process}\label{sec:model}

In accordance with the jargon of the P\'olya urn model, we will refer to the agents that engage in a competition as \emph{colors}.
Consider two colors, labelled 1 and 2. Each color is associated with a positive \emph{fitness} value that reflects its intrinsic competitiveness. Let $f_i$ denote the fitness of color $i$, $i=1,2$, and $r = f_1/f_2$ the fitness ratio. Without loss of generality, we assume that $f_1 \geq f_2$ and hence $r \geq 1$. 

The resource that the agents compete for, which is measured in discrete units, will be generically referred to as balls. The competition starts at time $t=0$ with color $i$ having $x_{0i}$ balls, $i=1,2$. We consider a discrete-time process.  At each time step, one ball of one of the colors is added to the
system. Denote by $X_{i}(t)$ the number of balls with color $i$ at time $t$ and $\X(t) = (X_1(t), X_2(t))$.
The trajectory of the competition $\X=\{\X(t)\}_{t\in \N}$ then forms a discrete-time discrete-space stochastic process. The state space is the first quadrant of the integral lattice  $\N^2$.
The initial condition is $\X(0) =  \x_0 \triangleq (x_{01},x_{02})$. 

In a nonlinear P\'olya urn process with fitness, the ball added at time $t+1$ has color $i$ with probability
\[
p_{i}(t) = \frac{f_i X_{i}(t)^\beta}{f_1 X_{1}(t)^\beta + f_2 X_{2}(t)^\beta}.% = \frac{r X_{1}(t)^\beta}{r X_{1}(t)^\beta + X_{2}(t)^\beta}.
\]
Here $\beta\geq 0$ reflects the strength of the feedback by cumulative advantage. Note that the larger $\beta$ is, the stronger the feedback. When $\beta=0$, there is no feedback and the process falls back to a random walk (where the transition probabilities do not depend on $\X(t)$). 

More formally, the trajectory $\{\X(t)\}_{t\in \N}$ forms a Markov chain with initial condition $\X(0)=\x_0$ and stationary transition probabilities  $\P[\X(t+1)=\x' \mid \X(t)=\x]$
 given by
\begin{equation}\label{eq:transition-prob}
Q(\x,\x'; \beta,r) = \begin{dcases*}
\frac{rx_1^\beta}{rx_1^\beta+x_2^\beta}, & if $\x' = \x + (1,0)$;\\
\frac{x_2^\beta}{rx_1^\beta+x_2^\beta}, & if $\x' = \x + (0,1)$;\\
0,& otherwise.
\end{dcases*}
\end{equation}
We will call such a process a $(\beta,r,\x_0)$-urn process. 

The \emph{duration} and \emph{intensity} of a competition have been defined through events of ties in \cite{jiang2015competition}. We follow the same definitions here.
Given a 2D process $\X=\{\X(t)\}_{t\in \N}$, not necessarily an urn process introduced above, we say that a \emph{tie} occurs at time $t$ if $X_{1}(t) = X_{2}(t)$. For $n\geq 0$, let $T_n(\X)$ be the time of the $n$-th tie, defined recursively by \[
T_n(\X) = \inf\{t > T_{n-1}(\X): X_1(t)=X_2(t)\},\quad n\geq 1,
\]
where $T_0(\X) = -1$ by convention.
The {\em duration} $T(\X)$ of a competition is defined
to be the time of the last tie, i.e.,
\[
T(\X) = \sup\{T_n(\X):  n\geq 0,\, T_n(\X)< +\infty\}.
\]
Note that $T(\X)$ marks the end of the competition in the sense that there are no more ties  after this point in time, hence leaving one of the colors in the lead forever. 

Let $N_t$ be the number of ties up to time $t$, i.e. 
\[
N_t(\X) = \sum_{j=0}^t \ind{X_{1}(j) = X_{2}(j)}, 
\]
where $\ind{A}$ is the indicator of event $A$. 
The {\em intensity} $N(\X)$ of a competition is the total number of ties throughout the competition, i.e.,
\[
N(\X) = \lim_{t\to\infty} N_t(\X) = \sum_{t=0}^\infty \ind{X_{1}(t) = X_{2}(t)}, 
\]
This measures the intensity of the competition in the sense that it counts the number of potential changes in leadership.

When there is no confusion, we will also write $T$ for $T(\X)$, and similarly for $T_n$, $N_t$ and $N$. Note that  $N = N_T$, $T = T_N$, and that $T<+\infty$ if and only if $N<+\infty$.
With an abuse of notation, we will use $T(\beta,r,\x) = T(\beta,r,x_{1},x_{2})$ to denote $T(\X)$ for any $(\beta,r,\x)$-urn process $\X$, and similarly for $T_n$, $N_t$ and $N$. Throughout the rest of the paper, a boldfaced letter always has two components, e.g.~$\x=(x_1,x_2)$ and $\Y=(Y_1,Y_2)$. The notations such as $g(\x) = g(x_1,x_2)$ will be understood without mention.

%% file: TEX/related.tex
%%!TEX PS-program = pdflatexmk
%%!TEX root = ../main.tex

\subsection{Related Work}\label{sec:related}

Given the 90 years since the P{\'o}lya urn process was first introduced~\cite{Polya23}, it is not surprising that many different properties of this process have been characterized through rigorous mathematical treatment as well as simulations. Most work focuses on the so-called {\em market share}, i.e.~the fraction of balls in each color, for which convergence results and limit distributions have been established for different regimes of fitness or feedback strength, but rarely for both~\cite{Drinea:2002,oliveira2008balls,oliveira2009onset,zhu2009nonlinear,Walls12}. Other properties that have been studied more recently include the probability of ever taking the lead and the onset of monopoly~\cite{oliveira2009onset,Walls12}. When the feedback is superlinear ($\beta>1$), the winning color receives all but a finite number of balls, a phenomenon knowns as \emph{monopoly}, various aspects of which have been studied~\cite{oliveira2008balls,oliveira2009onset}. The metrics under investigation in this paper, duration and intensity, have been studied in \cite{jiang2015competition} for linear P\'olya urn process with fitness. 

The {\em Poissonization} \cite{Mahm08} and the exponential embeddding \cite{davis1990reinforced} are two major technical tools used in the study of P\'olya urn processes. Other methods are surveyed in \cite{Peman07,zhu2009nonlinear}. We will mainly follow the exponential embedding approach in the present work. We extend existing works by considering the effect of nonlinear CA on duration and intensity. The theoretical findings deepen our understanding of the interaction between fitness and feedback strength in CA competitions, which in turn sheds light on understanding applications that employ such models.

%% file: TEX/results.tex
%%!TEX PS-program = pdflatexmk
%%!TEX root = ../main.tex

\subsection{Overview of Results}\label{sec:results}

\prettyref{tab:results} summarizes our main results on the tail distributions of duration and intensity, which will be detailed in Sections~\ref{sec:duration} and \ref{sec:intensity}. We have used the standard notations of $O$, $\Theta$, and $\Omega$ in the table. In later sections, we will also use other standard notations such as $o$ and $\sim$ without further mention, where $g(x) \sim h(x)$ means $\lim_{x} g(x)/h(x)=1$ in the limiting process under consideration. 

To the best of our knowledge, all results related to nonlinear CA ($\beta \neq 1$) are new, with the exception of the case $\beta \leq 1/2$ and $r=1$. The linear case ($\beta=1$) is given in \cite{jiang2015competition}, and included here for completeness and comparison. 

The results are revealing and worth exploring. In the equal fitness case $r=1$, we observe a phase transition at $\beta = 1/2$. For $\beta\leq 1/2$, competitions never end~\cite{khanin2001probabilistic}. For $\beta > 1/2$, competitions always end, but can be very long and  intense, as both duration and intensity have power-law tails. 

The picture is dramatically different in the case of different fitnesses. In this case, if $\beta \leq 1$ then the fittest agent is bound to win the competition (i.e.,  take the lead forever)\footnote{For $\beta=1$, see \cite{Mahm08} for a proof. For $\beta < 1$, see the remark at the end of \prettyref{subsubsec:duration-diff-sublinear-proof}.}. If $\beta > 1$, then there is a nonzero probability that the less fit wins (it actually becomes the monopoly)~\cite{davis1990reinforced}. In the sublinear regime, the fittest color wins relatively quickly with the distribution of duration upper bounded by a Weibull tail. Thus fitness superiority brings a clear advantage in this regime, in sharp contrast to the superlinear regime. Note that not only may the competition duration increase when moving from the linear to superlinear regime depending on $r$ and $\x_0$, but the fittest may even lose the competition! Thus the fittest may have to struggle much more under superlinear CA. 

Also observe that moving from equal to non-equal fitness induces longer competitions under superlinear CA. However, there is an advantage in becoming fitter since the chance of winning is greater than under equal fitness (where the chance is 50\% if $x_{01} = x_{02}$), but at the expense of engaging in potentially longer competitions. In a nutshell, superlinear CA may exacerbate the strugle of the fittest! 

Finally, in the case of different fitnesses, competition intensity is always small, exhibiting an exponential tail. This phenomenon of long (duration) but mild (intensity) competitions has been observed in~\cite{jiang2015competition}. We observe here that this phenomenon persists in the presence of nonlinear CA. 

\begin{table}[t]
\renewcommand{\arraystretch}{1.3} 
\centering
\begin{tabular}{ccccc}
\toprule
   &  \multicolumn{2}{c}{$\P[T(\beta,r,\x_0)\geq t]$}  &   \multicolumn{2}{c}{$\P[N(\beta,r,\x_0)\geq n]$}   \\
   \cmidrule{2-5}
   & $r=1$ & $r>1$ & $r=1$ & $r>1$\\
\midrule
$0\leq \beta \leq \frac{1}{2}$ & 1 & $e^{-\Omega\left(t^{1-\beta}\right)}$ & 1 & $O(a^n)$ \\
\midrule
$\frac{1}{2}<\beta<1$ & $\Theta(t^{\frac{1}{2}-\beta})$ & $e^{-\Omega\left(t^{1-\beta}\right)}$   & $\Omega(n^{-\beta})$ & $O(a^n)$ \\
\midrule
$\beta=1$ & $\Theta(t^{-\frac{1}{2}})$ & $\Omega(t^{(1-r)x_{01}})$  & $\Theta(n^{-1})$ & $O(a^n)$ \\
\midrule
 $\beta>1$ & $\Theta(t^{\frac{1}{2}-\beta})$ & $\Theta(t^{1-\beta})$  & $O(n^{-\beta})$  & $O(a^n)$ \\
\bottomrule
\end{tabular}
\caption{Tail distributions of duration and intensity. Here $a = 2/(r+1)$.}
\label{tab:results}
\end{table}

%% file: TEX/social_tagging.tex
%!TEX root = ../main.tex

\subsection{Recent Application to Social Tagging}
\label{sec:social_tagging}

In this section we briefly describe an example of the applications of P{\'o}lya urns in computer science. Such applications could potentially leverage a more general model that incorporates nonlinear CA and fitness. By providing a theoretical understanding of duration and intensity we prepare the ground for the application of more general models. 

Social or collaborative tagging refers to the increasingly common process where users {\em tag} resources within online services~\cite{gupta2010survey,wagner2014semantic}. For example, users can bookmark a URL on Delicious\footnote{\small delicious.com, www.flickr.com, www.twitter.com}, annotate pictures on Flickr$^1$, and use hashtag to mark tweets on Twitter$^1$. An important consideration in this context is the dynamics behind tag generation and tag accumulation by the various resources such as URLs, pictures and tweets. In particular, a cumulative advantage effect (i.e., preferential attachment) has been empirically observed in social tagging in the sense that, as resources accumulate more tags, they tend to accumulate even more tags. In order to capture this phenomenon, models that embody cumulative advantage such as P{\'o}lya urn and Yule-Simon process have been used to represent how objects accumulate tags~\cite{cattuto2007semiotic,golder2006usage}. Models that also capture the inherent difference between tags, which can be interpreted as tag fitness \cite{halpin2007complex}, and models that leverage tag ranking to assess tag dynamics \cite{wagner2014semantic} have also been proposed. 

To illustrate such modeling within our framework, consider two URLs competing for bookmarks by users on Delicious, as presented and evaluated in~\cite{golder2006usage}. For $i=1,2$, let $f_{i}$ denote the intrinsic fitness of URL$_i$, and $X_{i}(t)$ the number of bookmarks that it has received by time $t$. When the CA feedback has strength $\beta > 0$, how will the two URLs accumulate bookmarks? Will the fittest URL emerge as the unchallenged winner? How many bookmarks will they have accumulated together when this occurs? 

An important consideration is the effectiveness of social tagging in describing and assessing online resources~\cite{wagner2014semantic}. For example, can poor quality URLs be overridden by late coming higher quality URLs in the bookmark competition? The answer to such questions depends fundamentally on the nature of the competition, as defined by the fitnesses $f_{1,2}$ and the feedback strength $\beta$. Our work provides a solid theoretical ground for understanding such behaviors. For example, we now know that under superlinear CA much longer competitions can occur (in comparison to linear CA), as well as the fittest losing the competition. Such findings may put into question the effectiveness of social tagging.

%% file: TEX/ordering.tex
%%!TEX PS-program = pdflatexmk
%%!TEX root = ../main.tex

\section{Stochastic Ordering Results}\label{sec:order}

In this section, we will show that some of the metrics introduced in \prettyref{sec:model} can be ordered stochastically according to the feedback strength $\beta$.
 We recall the following definition of \emph{stochastic dominance}.

\begin{definition}[Stochastic dominance]
A random \\ variable $Z_1$ \emph{stochastically dominates} a random variable $Z_2$, if $\P[Z_1\geq z] \geq \P[Z_2\geq z]$ for all $z$. This is denoted by $Z_1\succeq Z_2$ or $Z_2\preceq Z_1$.
\end{definition}

\subsection{Equal Fitness} \label{subsec:dominance-equal}

The following theorem shows that in the equal fitness case, stronger feedback, i.e.~larger $\beta$, leads to stochastically shorter and less intense competitions.

\begin{theorem}\label{thm:dominance-equal}
Let $\beta\geq \beta'\geq 0$. The following hold,
\begin{itemize}
\item[$(i)$] $N_t(\beta,1,\x_0)\preceq N_t(\beta',1,\x_0)$ for all $t$;
\item[$(ii)$] $N(\beta,1,\x_0)\preceq N(\beta',1,\x_0)$;
\item[$(iii)$] $T_n(\beta,1,\x_0)\succeq T_n(\beta',1,\x_0)$ for all $n$;
\item[$(iv)$] $T(\beta,1,\x_0)\preceq T(\beta',1,\x_0)$.
\end{itemize}
\end{theorem}

%Note that $(iii)$ and $(iv)$ seem paradoxical. However, $(iii)$ implies that $T_n(\beta,1,\x_0)$ is more likely to be infinite than $T_n(\beta',1,\x_0)$, which suggests that the competition is more likely to end early,  consistent with $(iv)$.

\begin{proof}%[of \prettyref{thm:dominance-equal}]
Let $\X$  be a $(\beta,1,\x_0)$-urn process and let $\X'$ be a $(\beta',1,\x_0)$-urn process. Define a new process $\Y$ by $Y_1(t) = \min\{X_1(t), X_2(t)\}$ and $Y_2(t) = \max\{X_1(t),X_2(t)\}$. Similarly, define $\Y'$ by $Y'_1(t) = \min\{X'_1(t), X'_2(t)\}$ and $Y'_2(t) = \max\{X'_1(t),X'_2(t)\}$. 

Let $\{\eta_j\}_{j\in \N}$ be a sequence of independent random variables uniformly distributed on $[0,1]$. Define $\{ \Z(t)\}_{t\in \N}$ recursively by  $Z_1(0) = \min\{x_{01},x_{02}\}$, $Z_2(0) = \max\{x_{01}, x_{02}\}$, and
\begin{align*}
 Z_1(t+1) &= Z_1(t) +  \ind{Z_1(t)< Z_2(t),\eta_t\leq \frac{Z_1(t)^\beta}{Z_1(t)^\beta +  Z_2(t)^\beta} },\\
 Z_2(t+1) &=  Z_1(t) + Z_2(t) + 1 - Z_1(t+1),
\end{align*}
Define  $\{ \Z'(t)\}_{t\in\N}$ by the same equations but with $\beta$ replaced by $\beta'$.  Note that $\Y\overset{d}{=} \Z$ and $\Y'\overset{d}{=} \Z'$, where $\overset{d}{=}$ means ``equal in distribution''. It is also clear that $Z_1(t)\leq Z_2(t)$ for all $t$.

We now show that $Z_1(t) \leq Z'_1(t)$ by induction on $t$.
The base case $t=0$ holds trivially. Assume it holds for $t$ and consider $t+1$. Note that $Z_1(t) + Z_2(t)  = Z'_1(t)+Z'_2(t)= x_{01} + x_{02} + t$ and hence $Z_2'(t) \leq Z_2(t)$. There are three cases.
\begin{itemize}
\item
$Z_2'(t) = Z_1'(t)=Z_1(t)$. In this case, $Z_2(t) = Z_1(t)$, and hence $Z_1(t+1) = Z_1(t) \leq Z'_1(t) = Z'_1(t+1)$

\item $Z_2'(t) = Z_1'(t)\geq Z_1(t) + 1$. In this case, $Z_1(t+1) \leq Z_1(t) +1 \leq Z'_1(t) = Z'_1(t+1)$

\item $Z_2'(t) > Z_1'(t)$. In this case, $Z_2(t)\geq Z_2'(t) > Z_1'(t) \geq Z_1(t)$. Thus
 \[
 \frac{Z_1(t)^\beta Z'_2(t)^{\beta'}}{Z_2(t)^\beta Z'_1(t)^{\beta'}} = \left(\frac{Z_1(t)}{Z_2(t)}\right)^{\beta-\beta'} \left(\frac{Z_1(t) }{Z'_1(t) }\cdot \frac{ Z'_2(t)}{ Z_2(t)}\right)^{\beta'}\leq 1,
 \]
and hence
\[
\frac{Z_1(t)^\beta}{Z_1(t)^\beta +  Z_2(t)^\beta} \leq \frac{Z'_1(t)^{\beta'}}{Z'_1(t)^{\beta'} +  Z'_2(t)^{\beta'}}.
\]
It follows that
\begin{align*}
 &\quad\ Z_1(t+1) = Z_1(t) + \ind{\eta_t\leq \frac{Z_1(t)^\beta}{Z_1(t)^\beta +  Z_2(t)^\beta} }\\
&\leq   X'_1(t) + \ind{\eta_t\leq \frac{Z'_1(t)^{\beta'}}{Z'_1(t)^{\beta'} +  Z'_2(t)^{\beta'}}} =  Z'_1(t+1).
\end{align*}

\end{itemize}
In all cases, we have $Z_1(t+1)  \leq   Z'_1(t+1)$, which completes the induction. As a consequence,
\[
Z_2(t) -  Z_1(t) \geq Z'_2(t) -  Z'_1(t) \geq 0, \quad \forall t\geq 0.
\]
Thus $\Z$ ties at $t$ only if $\Z'$ also ties at $t$, which implies $N_t(\Z) \leq N_t(\Z')$, $N(\Z) \leq N(\Z')$, $T_n(\Z) \geq T_n(\Z')$, and $T(\Z) \leq T(\Z')$. 
 
Note that $N_t(\X) = N_t(\Y) \overset{d}{=} N_t(\Z)$ and $N_t(\X') = N_t(\Y') \\  \overset{d}{=} N_t(\Z')$, from which $(i)$ follows. The same argument also proves $(ii)$, $(iii)$ and $(iv)$. Alternatively, $(ii)$ follows from $(i)$ by letting $t\to\infty$, while $(iii)$ follows from $(i)$ by the identity $\{T_n\geq t\} = \{N_t\leq n\}$. 
\myfill \end{proof}

\subsection{Different Fitnesses}\label{subsec:dominance-diff}

In the case of different fitnesses, there are no such nice ordering results as in \prettyref{subsec:dominance-equal}, as we will see in \prettyref{fig:T-diff-sublinear} of \prettyref{subsec:duration-diff}. However, we  have some partial results, which will be useful later in characterizing the tail distributions of duration and intensity. Note that the results apply to the equal fitness case as well.

The following theorem shows that the time of first tie can be ordered stochastically. The proof uses a coupling argument similar to the one in the proof of \prettyref{thm:dominance-equal} and is found in  \prettyref{app:dominance-diff}.

\begin{theorem}\label{thm:dominance-diff}\label{THM:DOMINANCE-DIFF}
Let $\beta\geq \beta'\geq 0$. $T_1(\beta, r,\x_0)\succeq T_1(\beta',r',\x_0')$, %There exists a $(\beta,r,\x)$-urn process $\X$ and a $(\beta',r',\x')$-urn process $\X'$ such that $T_1(\X) \geq T_1(\X')$, 
if either of the following conditions holds,
\begin{itemize}
\item[$(i)$]
$r\geq r'$ and $x_{01} \geq x_{01}'\geq x_{02}'\geq x_{02}$;
\item[$(ii)$]
$r= r'$ and $x_{01} \leq x_{01}'\leq x_{02}'\leq x_{02}$.
\end{itemize}
In particular, $T_1(\beta,r,\x_0)\succeq T_1(\beta',r,\x_0)$.
\end{theorem}

When competition starts out with a tie, $T_1(\beta,r,x_0,x_0) = T_1(\beta',r,x_0,x_0) = 0$ trivially. What is more interesting in this case is the time of the first return to a tie, which can also be ordered as shown by the next corollary.

\begin{corollary} \label{corollary:first-return}
$T_2(\beta,r,x_0,x_0)\succeq T_2(\beta',r,x_0,x_0)$, if  $\beta\geq \beta'$.
In particular, the probability of ever tying again satisfies
\[
\P[T_2(\beta,r,x_0,x_0)<\infty] \leq \P[T_2(0,r,0,0)<\infty] = \frac{2}{r+1}.
\]
\end{corollary}

\begin{proof}
Let $p_r = r/(r+1)$ and $q_r = 1/(r+1)$.
By considering the first transition and applying \prettyref{thm:dominance-diff}, we obtain
\begin{align*}
&\quad\ \P[T_2(\beta,r,x_0,x_0)\geq t] \\
&= p_r \P[T_1(\beta,r,x_0+1,x_0)\geq t] + q_r \P[T_1(\beta,r,x_0,x_0+1)\geq t]\\
&\geq p_r \P[T_1(\beta',r,x_0+1,x_0)\geq t]  + q_r \P[T_1(\beta',r,x_0,x_0+1)\geq t]\\
&=\P[T_2(\beta',r,x_0,x_0)\geq t].
\end{align*}
which means $T_2(\beta,r,x_0,x_0)\succeq T_2(\beta',r,x_0,x_0)$.
In particular,
\begin{align*}
\P[T_2(\beta,r,x_0,x_0)<\infty] &\leq \P[T_2(0,r,x_0,x_0)<\infty]\\
&= \P[T_2(0,r,0,0)<\infty] = \frac{2}{r+1},
\end{align*}
where we have used the translation invariance of random walks  and the well-known formula for the probability of no return to the origin (see e.g. Section XI.3.c of \cite{fellerOne}). \myfill
\end{proof}

The next corollary shows that feedback, regardless of its strength $\beta$, does not increase competition intensity. In particular, competition always ends if $r>1$.

\begin{corollary}\label{corollary:intensity-diff}
$N(\beta,r,\x_0)\preceq N(0,r,\x_0)$ for any $\beta\geq 0$.
\end{corollary}

\begin{proof}
Let $\X$ be a $(\beta,r,\x_0)$-urn process. Let $F_n(z) = \P[X_1(T_n(\X))=z\mid T_n(\X)<\infty]$. Note that $T_n(\X)$ is a stopping time of $\X$ for $n\geq 1$. The strong Markov property and \prettyref{corollary:first-return} yield
\begin{align*}
&\quad \P[T_{n+1}(\X) < \infty\mid T_{n}(\X)<\infty]\\
&=\sum_{z} F_n(z)\P[T_2(\beta,r,z,z)<\infty] \\
&\leq \sum_{z} F_n(z) \P[T_2(0,r,0,0)<\infty] =\P[T_2(0,r,0,0)<\infty].
\end{align*}
Therefore,
\begin{align*}
&\quad\ \P[N(\beta,r,\x_0)\geq n] = \P[T_n(\X) < \infty]\\
&= \P[T_1(\X)<\infty] \prod_{j=1}^{n-1} \P[T_{j+1}(\X)<\infty\mid T_{j}(\X)<\infty]\\
&\leq \P[T_1(0,r,\x_0)<\infty] \left(\P[T_2(0,r,0,0)<\infty]\right)^{n-1} \\
&= \P[N(0,r,\x_0)\geq n],
\end{align*}
which means $N(\beta,r,\x_0)\preceq N(0,r,\x_0)$.
\myfill \end{proof}

%% file: TEX/duration.tex
%%!TEX PS-program = pdflatexmk
%%!TEX root = ../main.tex

\section{Tail Distribution of Duration}\label{sec:duration}

In this section, we characterize the tail distribution of duration $T$. The analysis relies on Rubin's exponential embedding that appeared in the appendix of \cite{davis1990reinforced}. We first review the exponential embedding in \prettyref{subsec:embedding}. We then present the tail distribution of $T$ for the case $r=1$ in \prettyref{subsec:duration-equal} and that for the case $r>1$ in \prettyref{subsec:duration-diff}.

\subsection{The Exponential Embedding}\label{subsec:embedding}

Rubin's exponential embedding is a specific representation of an urn process. 
Let $\{\xi_{kj}:k\in \{1,2\}, j\in \N\}$ be a set of independent exponential random variables with $\E\xi_{kj} = f_k^{-1} j^{-\beta}$, where $f_k$ is the fitness of color $k$. Let
\[
S_{k}(x,y) = \sum_{j=x}^{y-1} \xi_{kj},
\]
where by convention the sum is zero if $y\leq x$.  Given $\x_0$, order $\{S_{k}(x_{0k},x_k): x_k>x_{0k}, k\in \{1,2\}\}$ in increasing order and let $\tau_1<\tau_2<\dots$ be the resulting sequence. Let
\begin{equation}\label{eq:embed}
X_{k}(t) = \sup \{x\in \N: S_{k}(x_{0k},x) \leq \tau_t \}.
\end{equation}
Note that $S_{k}(x_{0k},x)$ can be considered as the time when color $k$ gets its $x$-th ball, and $X_{k}(t)$ is the number of balls with color $k$ when the total number of new balls arriving after time zero is $t$. The following theorem asserts that the process $\X$ constructed above is a $(\beta,r,\x_0)$-urn process. 

\begin{theorem}[Rubin]\label{thm:rubin}
The process $\{\X(t)\}_{t\in \N}$ defined by \prettyref{eq:embed} is a $(\beta,r,\x_0)$-urn process, where $r=f_1/f_2$.
\end{theorem}

We will use this representation throughout the rest of \prettyref{sec:duration}. Without further mention, $\{\xi_{kj}\}$ will always denote the set of independent random variables in this representation and $S_k$ the associated partial sums. We will also use the following notation, 
\begin{equation}\label{eq:Delta}
\Delta(\x, \y) = \Delta(x_{1},x_{2},y_1,y_2)= S_1(x_{1},y_1) - S_2(x_{2},y_2).
\end{equation}

When $f_k=1$, the characteristic function of $S_k(x,y)$ is given by 
\begin{equation}\label{eq:Psi}
\Psi(s;\beta,x,y) = \prod_{j=x}^{y-1} \left(1-\frac{is}{j^{\beta}}\right)^{-1}.
\end{equation}

The quantity $K(\beta,1,\x_0)$ defined in the following lemma will be used in the statements of the main results of the next two sections. Its proof is found in \prettyref{app:cf}.

\begin{lemma}\label{lemma:cf}\label{LEMMA:CF}
If either $(i)$ $\beta>1$, or $(ii)$ $\beta>1/2$ and $r=1$, then
\begin{equation}\label{eq:Psi-tilde}
\tilde \Psi(s;\beta,r,\x_0)\triangleq \lim_{x\to\infty}\Psi(s;\beta,x_{01},x) \Psi^*(rs;\beta,x_{02},x) 
\end{equation}
exists, and
\begin{equation}\label{eq:K}
K(\beta,r,\x_0) \triangleq \frac{1}{2\pi} \int_{-\infty}^\infty  \tilde \Psi(s;\beta,r,\x_0) ds
\end{equation}
is a strictly positive real number.
\end{lemma}

\subsection{Equal Fitness}\label{subsec:duration-equal}

We consider the equal fitness case in this section. Since the transition probability in \prettyref{eq:transition-prob} depends only on $r$, we assume  without loss of generality that $f_1 = f_2 = 1$ throughout this section. The main result is presented in \prettyref{subsec:duration-equal-result}. \prettyref{subsec:invariance} reviews the invariance principle, a key ingredient of the proof, which is given in \prettyref{subsec:duration-equal-proof}.

\subsubsection{Main Result}\label{subsec:duration-equal-result}

The following result was proved in \cite{khanin2001probabilistic} (see also Theorem 1 in \cite{oliveira2008balls}), from which it follows that $\P[T(\beta,1,\infty) \geq t] = 1$ for all finite $t$ and $\beta \in [0, 1/2]$.

\begin{theorem}[\cite{khanin2001probabilistic}]\label{thm:threshold}
With probability one, $T(\beta,1,\x_0)$ is finite if and only if $\beta>1/2$.
\end{theorem}

Our focus of this section is thus the regime $\beta > 1/2$. The  following theorem shows that $T(\beta,1,\x_0)$ has a power-law tail with exponent $\beta-1/2$ in this case, irrespective of the initial condition $\x_0$.

\begin{theorem}\label{thm:duration-equal}
For $\beta>1/2$, 
\begin{equation}\label{eq:duration-equal}
\P[T(\beta,1,\x_0)\geq t] \sim \frac{2^{\beta-1/2}}{\sqrt{(2\beta-1)\pi}} K(\beta,1,\x_0) t^{\frac{1}{2}-\beta}.
\end{equation}
\end{theorem}

\begin{figure}[t]
\centering
\includegraphics[width=\figwidth]{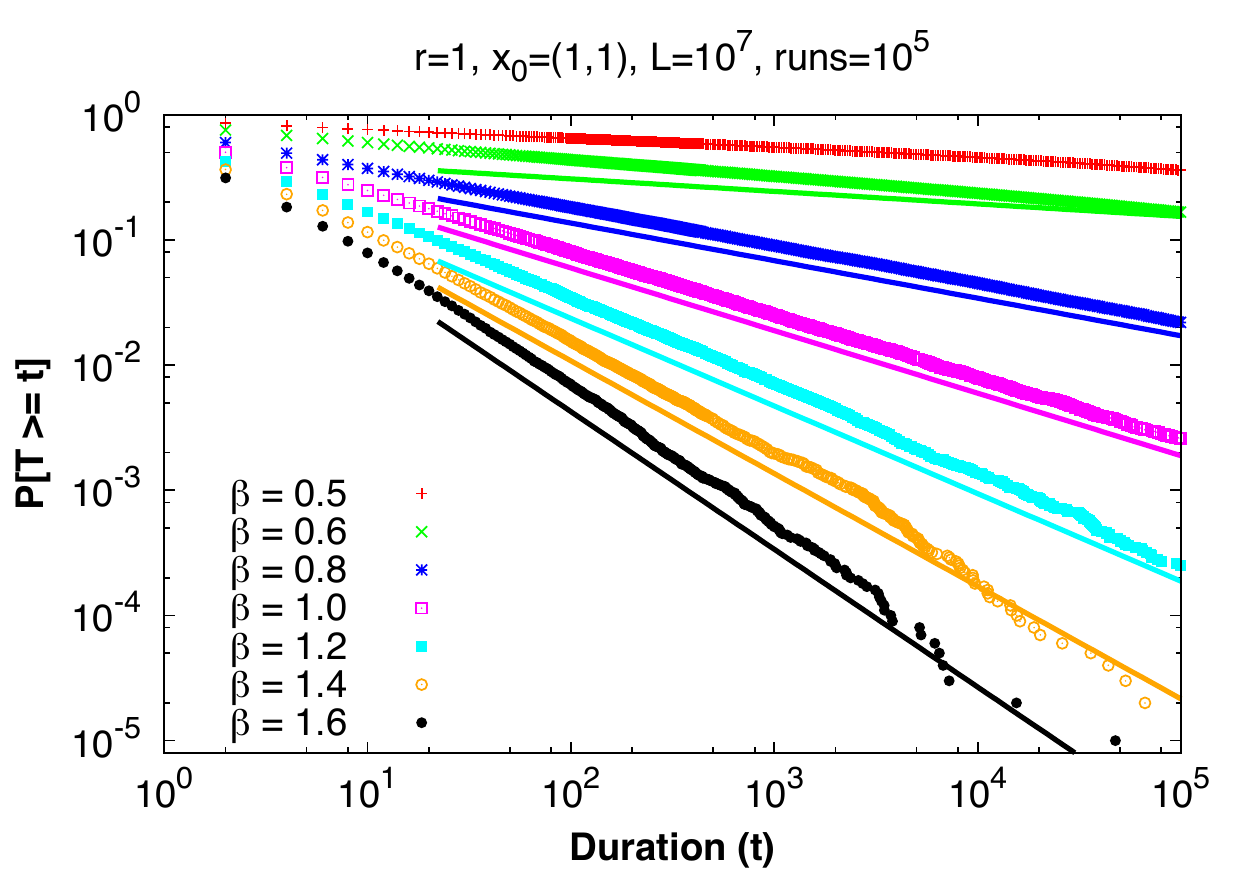}
\caption{
Tail distribution for duration of $r=1$ and various values of $\beta$. Dots (marks) are simulation results. The solid lines have slopes $1/2-\beta$.}
\label{fig:T-equal}
\end{figure}

The result is illustrated in \prettyref{fig:T-equal}, which shows the empirical tail distributions of duration from simulations. Each curve is obtained from $10^5$ independent runs of $L=10^7$ time steps each. The same simulation setup is used for all later plots and will not be repeated. Strictly speaking, what are plotted here are the tail distributions of the last tie before the simulation cutoff time $L$, which are good approximations to the true tail distributions $\P[T(\beta,1,\x_0)\geq t]$ for $t\ll L$.  Similar comments apply to later plots. We observe the stochastic ordering asserted by \prettyref{thm:dominance-equal}. \prettyref{fig:T-equal} also superimposes straight lines with slopes $1/2-\beta$, which are parallel to the asymptotes of \prettyref{eq:duration-equal}. Since we do not have a closed form formula for $K(\beta,1,\x_0)$, we have arbitrarily chosen the intercepts of these lines to ease comparison of their slopes with those of the simulated curves.  Note the good agreement between the corresponding slopes. Note also that for $\beta \leq 1/2$, the simulated tail distribution approaches the distribution $\P[T(\beta,1,\x_0)\geq t] = 1$, and dominates all curves for $\beta > 1/2$. In fact, this stochastic dominance result can be established by the same coupling argument used in the proof of \prettyref{thm:dominance-equal}.

\subsubsection{The Invariance Principle}\label{subsec:invariance}

In this section, we review a key ingredient of the proof of \prettyref{thm:duration-equal}, i.e.~the invariance principle, which asserts that an appropriately scaled random walk converges to a Wiener process in distribution. This has been exploited in the study of nonlinear P\'olya urn processes in \cite{oliveira2008balls}. We will follow a similar approach, but for our purpose, we will need not only the convergence result but also the rate of convergence, which is provided by the following result of Sakhanenko.

Let $\theta_1,\theta_2,\dots$ be a sequence of independent random variables with $\E\theta_j = 0$ and $\E\theta_j^2 < \infty$ for all $j$. Define a random process $\Xi_\theta$ with piecewise linear continuous sample paths by 
\[
\Xi_\theta(t) = \sum_{j=1}^\ell \theta_j + \frac{t-\sigma_\ell^2}{\E \theta_{\ell+1}^2} \theta_{\ell+1}, \quad \text{for } t\in [\sigma_\ell^2, \sigma_{\ell+1}^2], \ell = 0,1,\dots,
\]
where $\sigma_\ell^2 = \sum_{j=1}^\ell \E \theta_j^2$. Note that $\Xi_\theta(\sigma_\ell^2) = \sum_{j=1}^\ell \theta_j$.

%The following theorem bounds the error of approximating $\Xi_\theta$ by a Wiener process, which is Theorem A of \cite{sakhanenko2006estimates} with $b=1$.
%\begin{theorem}\label{thm:Sakhanenko}
%Let $\Xi_\theta$ be defined as above. For $\alpha\geq 2$,  there exists a Wiener process $W=W_{\alpha}$ such that
%\[
%\P\left[\sup_{t\in [0, L_2)} |\Xi_\theta(t)-W(t)|\geq C\alpha  y\right] \leq \frac{L_{\theta}^{\alpha}}{y^\alpha} + \P\left[\sup_j |\theta_j| > y\right],
%\]
%where $C$ is an absolute constant and 
%\[
%L_{\theta}^{\alpha} = \sum_{j=1}^\infty \E |\theta_j|^\alpha.
%\]
%\end{theorem}

  The following theorem, which is a special case of Theorem 1 of \cite{sakhanenko2006estimates}, bounds the error incurred by approximating $\Xi_\theta$ by a Wiener process.
\begin{theorem}[Sakhanenko]\label{thm:Sakhanenko}
Let $\Xi_\theta$ be defined as above. For $\alpha\geq 2$,  there exists a constant $\kappa$ and a Wiener process $W=W_{\alpha}$ such that for any $y>0$,
\begin{equation}\label{eq:Sakhanenko}
\P\left[\sup_{0\leq t < L_\theta^2} |\Xi_\theta(t)-W(t)|\geq 2\kappa \alpha  y\right] \leq \frac{L_{\theta}^{\alpha}}{y^\alpha},
\end{equation}
where  $L_{\theta}^{\alpha} = \sum_{j=1}^\infty \E |\theta_j|^\alpha$.
\end{theorem}

We now apply \prettyref{thm:Sakhanenko} to prove the following lemma, which is a key step in the proof of \prettyref{thm:duration-equal}.

\begin{lemma} \label{lemma:normal-approx}
Assume $\beta > 1/2$, $c>0$ and $\epsilon \in (0,c)$. If $x_m \sim m$ and $q_m=\Omega(\sqrt{m})$, then for all large enough $m$,
\begin{equation}\label{eq:normal-approx-upper}
\P \left[\sup_{y\geq x_m}  \Delta(x_m,x_m,y,y)  > \frac{c q_m}{m^{\beta}} \right]  \leq 2\bar\Phi\left(\frac{c^- q_m}{\sqrt{2m}} \right) + O(m^{-\beta}),
\end{equation}
and
\begin{equation}\label{eq:normal-approx-lower}
\P \left[\sup_{y\geq x_m}  \Delta(x_m,x_m,y, y)  > \frac{c q_m}{m^{\beta}} \right]  \geq 2\bar\Phi\left(\frac{c^+ q_m}{\sqrt{2m}} \right) - O(m^{-\beta}),
\end{equation}
where $\Delta$ is defined in \prettyref{eq:Delta},
\[
c^\pm = (c\pm\epsilon) \sqrt{2\beta-1},
\]
and $\bar \Phi$ is the CCDF of the standard normal distribution,
\[
\bar\Phi(z) = \frac{1}{\sqrt{2\pi}} \int_z^\infty e^{-u^2/2}du.
\]
\end{lemma}

\begin{proof}
Let $\theta_j = \xi_{1(j+x_m-1)} -\xi_{2(j+x_m-1)}$ for $j\geq 1$. Define $\Xi_\theta$ and $L_\theta^\alpha$ as in \prettyref{thm:Sakhanenko}. Note that \[
\sup_{y\geq x_m} \Delta(x_m,x_m,y,y) = \sup_{\ell\geq 0} \sum_{j=1}^{\ell} \theta_j =\sup_{0\leq t< L_\theta^2} \Xi_\theta(t).
\]
Thus
\[
E  \triangleq  \left\{\sup_{y\geq x_m} \Delta(x_m, x_m,y,y)  > \frac{c q_m}{m^{\beta}} \right\}=\left\{\sup_{0\leq t< L_\theta^2} \Xi_\theta(t) > \frac{c q_m}{m^{\beta}} \right\}.
\]
Let $W$ be the Wiener process in \prettyref{thm:Sakhanenko}, and
\begin{align*}
E_0 & = \left\{\sup_{0\leq t< L_\theta^2}  |\Xi_\theta(t)-W(t)| > \frac{\epsilon q_m}{2 m^\beta} \right\},\\
E_\pm & = \left\{\sup_{0\leq t< L_\theta^2}  W(t) > \left(c\pm\frac{1}{2}\epsilon\right) \frac{q_m}{m^\beta}\right\}.
%E_- & = \left\{\sup_{0\leq t< L_\theta^2}  W(t) > (c-\epsilon) q_n m^{-\beta} \right\}.
\end{align*}
Since $E_+\subset E\cup E_0$ and $E\subset E_-\cup E_0$, we have
\begin{equation}\label{eq:sandwich}
\P[E_+] - \P[E_0]\leq \P[E] \leq \P[E_-] + \P[E_0].
\end{equation}

We first show that $\P[E_0]=O(m^{-\beta})$. \prettyref{thm:Sakhanenko} yields
\[
\P[E_0] \leq  L_\theta^\alpha \left(\frac{4\alpha\kappa}{\epsilon}\right)^\alpha m^{\alpha\beta} q_m^{-\alpha}  \leq  L_\theta^\alpha \left(\frac{4\alpha\kappa}{\lambda\epsilon}\right)^\alpha m^{\alpha\beta-\alpha/2},
\]
where we have used $q_m \geq \lambda \sqrt{m}$ for some $\lambda>0$ in the last step.
Note that 
\[
|\xi_{1j}-\xi_{2j}|^\alpha \leq \max \{ \xi_{1j}^\alpha, \xi_{2j}^\alpha \} \leq \xi_{1j}^\alpha + \xi_{2j}^\alpha,
\]
and $\E \xi_{1j}^\alpha = \E \xi_{2j}^\alpha = \Gamma(\alpha+1) j^{-\alpha\beta}$, where $\Gamma(\cdot)$ is the gamma function. Thus for $\alpha > \beta^{-1}$,
\begin{align*}
L_\theta^\alpha &= \sum_{j=1}^\infty \E|\theta_j|^\alpha \leq 2 \sum_{j=x_m}^\infty \E \xi_{1j}^\alpha =  2 \Gamma(\alpha+1) \sum_{j=x_m}^\infty j^{-\alpha\beta} \\
&\sim 2 \Gamma(\alpha+1) \int_m^\infty z^{-\alpha\beta}dz=  \frac{2\Gamma(\alpha+1)}{\alpha\beta-1} m^{1-\alpha\beta}.
\end{align*}
Set $\alpha=2+2\beta$, which satisfies $\alpha>\beta^{-1}$ for $\beta>1/2$. It follows that for all large enough $m$, 
\begin{equation}\label{eq:small-error}
\P[E_0] = O(m^{1-\alpha/2}) = O(m^{-\beta}). %\leq \phi m^{1-\alpha/2} = \phi m^{-\beta},
\end{equation}
%where
%\[
%\phi = \frac{2\Gamma(\alpha+1)}{\alpha\beta-1}\left(\frac{4\alpha \kappa}{\lambda\epsilon}\right)^\alpha.
%\]
Now we compute $\P[E_\pm]$.
The well-known formula for the distribution of the maximum of a Wiener process (see (6.5.3) of \cite{resnick1992adventures}) yields
\begin{align}\label{eq:max-Wiener}
\P[E_\pm]  = 2\bar \Phi\left(\frac{(c\pm\epsilon/2) q_m m^{-\beta}}{\sqrt{L_\theta^2}}\right) = 2\bar\Phi\left(\frac{c_m^\pm q_m}{\sqrt{2m}} \right),
\end{align}
where
%\[
%c_m^{\pm} =  (c\pm\epsilon/2)\sqrt{2}m^{1/2-\beta}(L_\theta^2)^{-1/2}. 
%\]
\[
c_m^{\pm} = \frac{(c\pm\epsilon/2)\sqrt{2}m^{1/2-\beta}}{\sqrt{L_\theta^2}}. %\to (c\pm\epsilon)\sqrt{\frac{2\beta-1}{2}}.
\]
Note that $\E\theta_j^2 = 2\text{Var}[\xi_{1(j+x_m-1)}]=2(j+x_m-1)^{-\beta}$. Thus
\begin{align*}
L_\theta^2  & = \sum_{j=1}^\infty \E\theta_j^2 =  2\sum_{j=x_m}^\infty j^{-2\beta} \sim 2 \int_{m}^\infty z^{-2\beta}dz =  \frac{2}{2\beta-1} m^{1-2\beta},
\end{align*}
from which it follows that 
\[
c_m^+\to \left(c+\frac{1}{2}\epsilon\right) \sqrt{2\beta-1}< c^+,
\]
and hence $c_m^+ < c^+$ for large $m$. 
Similarly, $c_m^- > c^-$. Therefore, \eqref{eq:normal-approx-upper} and \eqref{eq:normal-approx-lower} follow from \eqref{eq:sandwich}, \eqref{eq:small-error}, \eqref{eq:max-Wiener}, and the monotonicity of $\bar \Phi$.
\myfill \end{proof}

Together with some large deviation results, \prettyref{lemma:normal-approx} immediately yields the following bounds on the probability of ever having a tie, which is what will be used directly in the proof of \prettyref{thm:duration-equal}. The proof of \prettyref{lemma:first-visit} is found in  \prettyref{app:first-visit}.

\begin{lemma}\label{lemma:first-visit}\label{LEMMA:FIRST-VISIT}
Suppose $|\rho(\x)| = \Omega(1)$, where 
\begin{equation}\label{eq:rho}
\rho(\x) = \frac{x_1-x_2}{\sqrt{\|\x\|_1}}=\frac{x_1-x_2}{\sqrt{x_1+x_2}}.
\end{equation}
For $\beta>1/2$ and $\epsilon>0$, the following inequalities hold,
\begin{align*}
\P[T_1(\beta,1,\x)<\infty] &\leq 2\bar\Phi\left(c_1 |\rho(\x)|\right) + O(\|\x\|_1^{-\beta}),\\
\P[T_1(\beta,1,\x)<\infty] &\geq 2\bar\Phi\left(c_2 |\rho(\x)|\right) - O(\|\x\|_1^{-\beta}),
\end{align*}
where $c_1 = (1-\epsilon) \sqrt{2\beta-1}$ and $c_2 = \sqrt{2\beta-1}$.
\end{lemma}

%\begin{lemma}\label{lemma:first-visit}\label{LEMMA:FIRST-VISIT}
%Suppose $|\rho(\x)| = o(\|\x\|_1^{(\beta\wedge 1)-1/2})$ and $|\rho(\x)| = \Omega(1)$, where $\beta>1/2$, $a\wedge b= \min\{a,b\}$, and
%\begin{equation}\label{eq:rho}
%\rho(\x) = \frac{x_1-x_2}{\sqrt{\|\x\|_1}}=\frac{x_1-x_2}{\sqrt{x_1+x_2}}.
%\end{equation}
%For any $\epsilon>0$, the following inequalities hold,
%\begin{align*}
%\P[T_1(\beta,1,\x)<\infty] &\leq 2\bar\Phi\left(c_1 |\rho(\x)|\right) + O(\|\x\|_1^{-\beta}),\\
%\P[T_1(\beta,1,\x)<\infty] &\geq 2\bar\Phi\left(c_2 |\rho(\x)|\right) - O(\|\x\|_1^{-\beta}),
%\end{align*}
%where $c_1 = (1-\epsilon) \sqrt{2\beta-1}$ and $c_2 = \sqrt{2\beta-1}$.
%\end{lemma}

\subsubsection{Proof of \prettyref{thm:duration-equal}}\label{subsec:duration-equal-proof}

Let $A_t$ be the set of states reachable at time $t$ by a $(\beta,1,\x_0)$-urn process, i.e. 
\[
A_t = \{\x\in \N^2: \|\x\|_1 = \|\x_0\|_1 + t, \ x_k \geq x_{0k} \text{ for } k=1,2\}.
\]
We will need the following lemma in the proof of \prettyref{thm:duration-equal}. %The proof of the lemma is found in \prettyref{app:difference}.

\begin{lemma}\label{lemma:difference}\label{LEMMA:DIFFERENCE}
Let $\X$ be a $(\beta,1,\x_0)$-urn process and $A_t(\delta) = \{\x\in A_t: |\rho(\x)|\leq \delta\}$, where $\rho(\x)$ is defined in \prettyref{eq:rho}. For $\beta>1/2$ and $\gamma < \beta\wedge 1-1/2$, where $a\wedge b = \min\{a,b\}$, we have, as $t\to\infty$,
\[
t^{\beta}\P[\X(t) = \x]  \to 2^{\beta+1} K(\beta,1,\x_0),
\]
uniformly for $\x\in A_t(t^\gamma)$, where $K(\beta,1,\x_0)$ is given by \prettyref{eq:K}.
\end{lemma}

\begin{proof}
Note that $\X(t)=\x$ if and only if color 1 gets its $x_1$-th ball before color 2 gets its $(x_2+1)$-st ball and at the same time color 2 gets its $x_2$-th before color 1 gets its $(x_1+1)$-st ball. Using the exponential embedding, this probability is given by
\[
 \P[\X(t) = \x] = \P[-\xi_{1 x_1} < \Delta(\x_0,\x) < \xi_{2 x_2}],
\]
where $\Delta$ is defined in \prettyref{eq:Delta}.
Let $\psi(s;\x)$ denote the characteristic function of $\Delta(\x_0,\x)$, i.e.
\[
\psi(s;\x)  = \Psi(s;\beta,x_{01},x_1) \Psi^*(s;\beta,x_{02},x_2),
\]
where $\Psi$ is given by \prettyref{eq:Psi}.
By the inversion formula, 
\begin{align*}
&\quad\;\P[-\xi_{1x_1} < \Delta(\x_0,\x)  < \xi_{2x_2}\mid \xi_{1x_1},\xi_{2x_2}] \\
&= \frac{1}{2\pi} \int_{-\infty}^\infty  \psi(s;\x) \frac{e^{is  \xi_{1x_1}} - e^{-is  \xi_{2x_2}}}{is} ds.
\end{align*}
Deconditioning and interchanging the order of integrations by Fubini's theorem, we obtain
\begin{align*}
&\quad\;\P[-\xi_{1x_1} < \Delta(\x_0,\x)  < \xi_{2x_2}]\\
&=  \frac{1}{2\pi} \int_{-\infty}^\infty  \psi(s;\x) \frac{\E\left[e^{is  \xi_{1 x_1}}\right] - \E\left[e^{-is \xi_{1 x_2}}\right]}{is} ds\\
&=  \frac{1}{2\pi} \int_{-\infty}^\infty  \psi(s;\x) \frac{1}{is}\left[\left(1-\frac{is}{x_1^\beta}\right)^{-1} - \left(1+\frac{is}{x_2^\beta}\right)^{-1}\right]ds\\
&=  \frac{x_1^{-\beta}+x_2^{-\beta}}{2\pi} \int_{-\infty}^\infty  \psi(s ;x_1+1,x_2+1) ds.
\end{align*}

The rest of the proof is relegated to \prettyref{app:uniform}. We only sketch it here. For $\x\in A_t(t^\gamma)$, we have $x_1 + x_2 = \|\x\|_1=\|\x_0\|_1+t \sim t$, and $|x_2-x_1| = |\rho(\x)|\sqrt{\|\x\|_1} = O(t^{\gamma+1/2}) = o(t^{1\wedge \beta})$. It follows that as $t\to\infty$,
\[
\frac{x_1}{t}\to \frac{1}{2}, \frac{x_2}{t} \to \frac{1}{2},
\]
and 
\[
\psi(s;x_1+1,x_2+1) \to \tilde \Psi(s;\beta,1,\x_0),
\]
uniformly for $\x\in A_t(t^\gamma)$, where $\tilde \Psi$ is defined in \prettyref{eq:Psi-tilde}.  The proof is then completed by letting $t\to\infty$ and applying the Dominated Convergence Theorem. 
\myfill \end{proof}

Now we prove \prettyref{thm:duration-equal}.
\begin{proof}[of \prettyref{thm:duration-equal}]
Let $\X$ denote a $(\beta,1,\x_0)$-urn process and $A_t$ the set of reachable states at time $t$ as defined above. Let $\lambda > 0$ and $\gamma \in (0, \beta\wedge 1-1/2)$. Let 
\[
A_t^1 = A_t(\lambda), \quad A_t^2 = A_t \setminus A_t(t^\gamma), \quad A_t^3 = A_t(t^\gamma) \setminus A_t(\lambda),
\]
where $A_t(\delta) = \{\x\in A_t: |\rho(\x)|\leq \delta\}$ as in \prettyref{lemma:difference} and $|A_t(\delta)|\sim \delta t^{1/2}$.
%\begin{align*}
%A_1 &= A_t(\lambda), \\ %\{\x\in A_t:|\rho(\x)|\leq \lambda\},\\
%A_2 &= A_t \setminus A_t(t^\gamma), \\%  \{\x\in A_t:|\rho(\x)| > t^\gamma\},\\
%A_3 &= A_t(t^\gamma)\setminus A_t(\lambda).% \{\x\in A_t:\lambda < |\rho(\x)|\leq t^\gamma\}.
%\end{align*}
Note that $\P[T(\X)\geq t] = \sum_{j=1}^3 P_j$, where
\[
P_j = \sum_{\x\in A_t^j}  \P[\X(t) = \x] \cdot \P[T_1(\beta,1,\x)<\infty].
\]

We first bound $P_1$.  Since $|A_t^1| \sim \lambda t^{1/2}$, by \prettyref{lemma:difference},
\[
P_1 \leq \sum_{\x\in A_t^1}  \P[\X(t) = \x] \sim  \lambda 2^{\beta+1} K(\beta,1,\x_0)   t^{1/2-\beta}.
\]

To bound $P_2$, let $\y_t = \arg\min_{\x\in A_t^2} \rho(\x)$.
By \prettyref{thm:dominance-diff}, $\P[T_1(\beta,1,\x)<\infty] \leq \P[T_1(\beta,1,\y_t)<\infty]$ for $\x\in A_t^2$. Thus
\begin{align*}
P_2 &\leq \P[T_1(\beta,1,\y_t)<\infty]\ \P[\X(t) \in A_t^2 ]\leq \P[T_1(\beta,1,\y_t)<\infty] \\
&\leq 2\bar\Phi(c_1 t^{\gamma}) + O(t^{-\beta}) = o(t^{1/2-\beta}),
\end{align*}
where the last inequality follows from \prettyref{lemma:first-visit}.

Now we bound $P_3$. By Lemmas \ref{lemma:first-visit} and \ref{lemma:difference},
\begin{align*}
P_3 &= \sum_{\x\in A_t^3}  \P[\X(t) = \x] \cdot \P[T_1(\beta,1,\x)<\infty]\\
&\sim  2^{\beta} K(\beta,1,\x_0)  t^{-\beta} \sum_{\x\in A_t^3} \P[T_1(\beta,1,\x)<\infty]\\
& \leq  2^{\beta} K(\beta,1,\x_0)  t^{-\beta}\sum_{\x\in A_t^3}  \bar\Phi(c_1 |\rho(\x)|) + O(t^{\gamma+1/2-2\beta})\\
& \sim  2^{\beta} K(\beta,1,\x_0)  t^{1/2-\beta} \int_{\lambda}^{t^\gamma} \bar\Phi(c_1 u)du + o(t^{-\beta})\\
& \leq  2^{\beta} K(\beta,1,\x_0)  t^{1/2-\beta} \int_{0}^{\infty} \bar\Phi(c_1 u)du + o(t^{-\beta})\\
& =  \frac{2^{\beta-1/2}K(\beta,1,\x_0) }{(1-\epsilon)\sqrt{(2\beta-1)\pi}}   t^{1/2-\beta}  + o(t^{-\beta}),
\end{align*}
where we have used  $\int_0^\infty \bar \Phi(c_1 u)du = c_1^{-1}(2\pi)^{-1/2}$ in the last step.
Letting $t\to \infty$, we obtain from the bounds on the $P_j$'s,
\[
\limsup_{t\to\infty} \frac{\P[T(\X)\geq t]}{t^{1/2-\beta}} \leq  \lambda 2^{\beta+1} K(\beta,1,\x_0) + \frac{2^{\beta-1/2}K(\beta,1,\x_0)}{(1-\epsilon)\sqrt{(2\beta-1)\pi}} .
\]
Letting $\lambda,\epsilon \to 0$,
\begin{equation}\label{eq:duration-equal-upper}
\limsup_{t\to\infty} \frac{\P[T(\X)\geq t]}{t^{1/2-\beta}} \leq  \frac{2^{\beta-1/2}}{\sqrt{(2\beta-1)\pi}} K(\beta,1,\x_0).
\end{equation}
On the other hand,
\begin{align*}
P_3 &\sim  2^{\beta} K(\beta,1,\x_0)  t^{-\beta} \sum_{\x\in A_t^3} \P[T_1(\beta,1,\x)<\infty]\\
& \geq  2^{\beta} K(\beta,1,\x_0)  t^{-\beta}\sum_{\x\in A_t^3}  \bar\Phi(c_2 |\rho(\x)|) - O(t^{\gamma+1/2-2\beta})\\
& \sim  2^{\beta} K(\beta,1,\x_0) t^{1/2-\beta} \int_{\lambda}^{t^\gamma} \bar\Phi(c_2 u)du - o(t^{-\beta}).
\end{align*}
Since $\P[T(\X)\geq t] \geq P_3$, letting $t\to \infty$, we obtain,
\[
\liminf_{t\to\infty} \frac{\P[T(\X)\geq t]}{t^{1/2-\beta}} \geq  2^{\beta}   K(\beta,1,\x_0)   \int_{\lambda}^{\infty} \bar\Phi(c_2 u)du.
\]
Letting $\lambda\to 0$ and using $\int_0^\infty \bar \Phi(c_2 u)du = c_2^{-1} (2\pi)^{-1/2}$,
\begin{equation}\label{eq:duration-equal-lower}
\liminf_{t\to\infty} \frac{\P[T(\X)\geq t]}{t^{1/2-\beta}} \geq  \frac{2^{\beta-1/2}}{\sqrt{(2\beta-1)\pi}} K(\beta,1,\x_0).
\end{equation}
Combining \prettyref{eq:duration-equal-upper} and \prettyref{eq:duration-equal-lower} yields \prettyref{eq:duration-equal}.
\myfill \end{proof}

\subsection{Different Fitnesses}\label{subsec:duration-diff}

We consider in this section the case of different fitnesses. When the feedback is linear ($\beta=1$), it has been shown in \cite{jiang2015competition} that the duration has a power-law tail with exponent between $(r-1)x_{01}$ and $(r-1)(x_{01}-r^{-1})$. We focus on the superlinear ($\beta>1$) and sublinear ($\beta<1$) regimes in this section. The main results are presented in \prettyref{subsec:duration-diff-result}. The proof for the superlinear linear regime is given in \prettyref{subsubsec:duration-diff-superlinear-proof}, and that for the sublinear regime is given in \prettyref{subsubsec:duration-diff-sublinear-proof}.

\subsubsection{Main Results}\label{subsec:duration-diff-result}

The following theorem shows that when the feedback is superlinear, the duration $T(\beta,r,\x_0)$ has a power-law tail with exponent $\beta-1$. Compared to the duration $T(\beta,1,\x_0)$ with the same $\beta$ and $\x_0$ in the equal fitness case, the duration $T(\beta,r,\x_0)$ with $r>1$ has a significantly heavier tail, which means that in the superlinear regime competitions may become much longer when agents have different fitnesses, similar to the observation in \cite{jiang2015competition} for the linear regime. In contrast to the linear regime, however, the exponent in the superlinear regime does not depend on either the fitness ratio $r$ or the initial condition $\x_0$. %We should mention that by a theorem in \cite{davis1990reinforced}, the fitter agent may lose in the superlinear regime, which never occurs in the linear regime. 

\begin{theorem}\label{thm:duration-diff-superlinear}
For $r>1$ and $\beta>1$,
\begin{equation}\label{eq:duration-diff-superlinear}
\P[T(\beta,r,\x_0)\geq t] \sim t^{1-\beta}\frac{(r-1)2^{\beta-1}}{\beta-1} K(\beta, r,\x_0).
\end{equation}
\end{theorem}

If feedback is sublinear, however, $T(\beta,r,\x_0)$ no longer has  a power-law tail. As the following theorem shows, the tail distribution of $T(\beta,r,\x_0)$ is upper bounded by a Weibull distribution with shape parameter $1-\beta$. Thus in the sublinear regime, $T(\beta,r,\x_0)$ with $r>1$ always has a lighter tail than the corresponding $T(\beta,1,\x_0)$. In particular, when $\beta=0$, we recover the known exponential tail of $T(0,r,\x_0)$.  

\begin{theorem}\label{thm:duration-sublinear-diff}
For $r>1$ and $\beta < 1$,
%\begin{equation}\label{eq:duration-sublinear-diff-log-pmf}
%\limsup_{m\to\infty} \frac{\log \P[T(\beta,r,\x_0)=t_{2m}]}{t_{2m}^{1-\beta}} \leq \frac{1-r}{1-\beta} 2^{\beta-1} x_{01}^\beta.
%\end{equation}
%and
\begin{equation}\label{eq:duration-sublinear-diff-log}
\limsup_{t\to\infty} \frac{\log \P[T(\beta,r,\x_0)\geq t]}{t^{1-\beta}} \leq \frac{1-r}{1-\beta} 2^{\beta-1} x_{01}^\beta.
\end{equation}
\end{theorem}

\begin{figure}[t]
\centering
\begin{subfigure}{1\columnwidth}
\centering
\includegraphics[width= \figwidth]{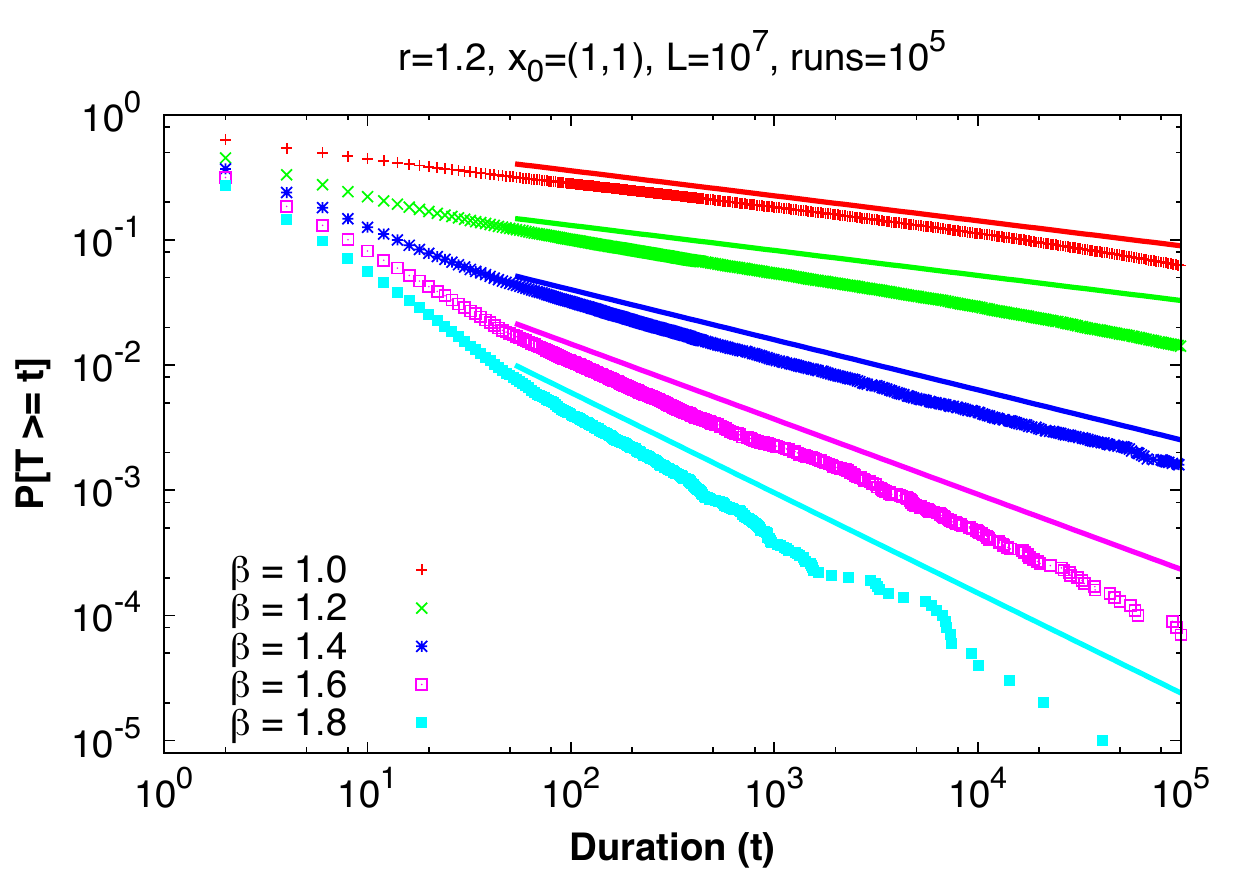}
\caption{$\beta\geq1$}
\label{fig:T-diff-superlinear} 
\end{subfigure}
\begin{subfigure}{1\columnwidth}
\centering
\includegraphics[width= \figwidth]{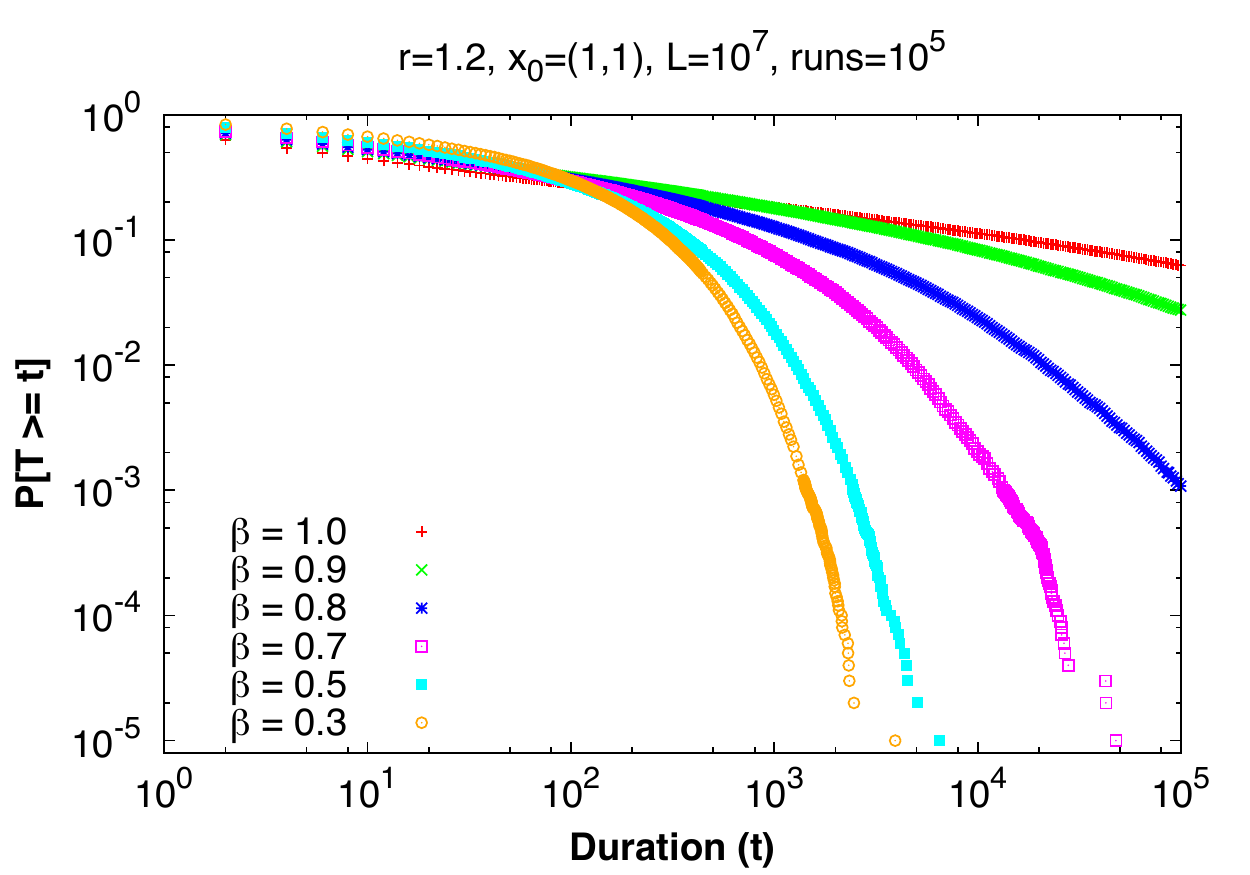}
\caption{$\beta\leq 1$}
\label{fig:T-diff-sublinear} 
\end{subfigure}
\caption{
Tail distribution of duration for $r=1.2$ and various values of $\beta$. Dots (marks) are simulation results. The solid lines in \subref{fig:T-diff-superlinear} have slopes $1-\beta$.}
\label{fig:T-diff}
\end{figure}

The results are illustrated in \prettyref{fig:T-diff}, which shows the simulated tail distributions of duration for $r=1.2$ and various $\beta$ values.  \prettyref{fig:T-diff-superlinear} shows the superlinear regime ($\beta \geq 1$). The power-law exponents from simulations are close to the theoretical values, though the agreement is not as good as in the equal fitness case, as the finite cutoff in simulation time has a greater impact here. Note that the curves for $\beta=1$ and $\beta=1.2$ are approximately parallel.  This is not a coincidence. For $\beta=1.2$, \prettyref{thm:duration-diff-superlinear} shows that the tail exponent is $\beta-1=0.2$. For $\beta=1$, \cite{jiang2015competition} shows that the tail exponent is roughly $(r-1)x_{01} = 0.2$. More generally, $T(\beta,r,\x_0)$ with $\beta>1$ may have a heavier or lighter tail than $T(1,r,\x_0)$, depending on $r$ and $\x_0$.

\prettyref{fig:T-diff-sublinear} shows the sublinear regime ($\beta \leq 1$). As mentioned in \prettyref{subsec:dominance-diff}, the crossover between the curves indicates that there is no simple stochastic ordering between $T(\beta,r,\x_0)$ of different $\beta$. However, the tails are still nicely ordered. Note that a larger $\beta$ results in a heavier tail, which is opposite to what we observe in the superlinear regime. When $\beta$ is small, the tail drops very fast.
Thus in the sublinear regime having the advantage of a larger fitness clearly manifests itself in shorter competition durations. 

\subsubsection{Proof of \prettyref{thm:duration-diff-superlinear}}\label{subsubsec:duration-diff-superlinear-proof}

Before we prove \prettyref{thm:duration-diff-superlinear}, we first prove the following result on the probability of never tying again when starting from a tie with a large number of balls. As a consequence of this result, for $r>1$ and large $t$, the probability of the duration being $t$ has the same order as the probability of having a tie at time $t$. 

\begin{lemma}\label{lemma:escape}
For $\beta\geq 0$ and $r\geq 1$, the probability of never tying again satisfies
\begin{equation}\label{eq:escape}
\lim_{x\to \infty} \P[T_2(\beta,r,x,x)=\infty] = \P[T_2(0,r,0,0)=\infty]=\frac{r-1}{r+1}.
\end{equation}
\end{lemma}

\begin{proof}
Note that for $\x\sim (x,x)$, the transition probability in \prettyref{eq:transition-prob} satisfies
\[
\lim_{x\to\infty} Q(\x,\x+\Delta\x;\beta,r) = Q(0,\Delta \x;0,r),
\]
which is the transition probability of a biased random walk. Thus for fixed $2k$, 
\[
\lim_{x\to\infty} \P[T_2(\beta,r,x,x)=2k] = \P[T_2(0,r,0,0)=2k].
\]
Since $\P[T_2(\beta,r,x,x) < \infty] =  \sum_{k=1}^\infty \P[T_2(\beta,r,x,x)=2k] $, Fatou's Lemma yields
\begin{align*}
&\liminf_{x\to\infty} \P[T_2(\beta,r,x,x) < \infty] \geq \sum_{k=1}^\infty \lim_{x\to\infty} \P[T_2(\beta,r,x,x) =2k]  \\
& \ = \sum_{k=1}^\infty  \P[T_2(0,r,0,0)=2k] = \P[T_2(0,r,0,0)<\infty].
\end{align*}
Now using \prettyref{corollary:first-return}, we obtain
\[
\lim_{x\to\infty} \P[T_2(\beta,r,x,x) < \infty]=\P[T_2(0,r,0,0)<\infty]=\frac{2}{r+1},
\]
which immediately implies \prettyref{eq:escape}.
\end{proof}

Now we prove \prettyref{thm:duration-diff-superlinear}.

%\subsubsection{Superlinear Case: $\beta>1$}

%\begin{theorem}\label{thm:tie_superlinear}
%For $\beta>1$ and $r\geq 1$, or $\beta>1/2$ and $r=1$,
%
%\end{theorem}
%
%\begin{proof}
%
%
%\begin{align*}
%\P[E_m] &=  \frac{r+1}{2\pi rm^{\beta}} \int_{-\infty}^\infty  \psi_{m+1}(s) ds.
%\end{align*}
%Taking the limit $m\to \infty$, 
%\[
%\lim_{m\to\infty} m^{\beta}\P[E_m] = \frac{r+1}{2\pi r} \int_{-\infty}^\infty   \psi_\infty(s) ds.
%\]
%\end{proof}

\begin{proof}[of \prettyref{thm:duration-diff-superlinear}]
Let $\X$ be a $(\beta, r, \x_0)$-urn process. Note that a tie occurs only at time epochs of the form $t_{2x} = 2x - \|\x_0\|_1$ for some integer $x$. At such a $t_{2x}$, both colors have $x$ balls. Note that
\begin{align*}
\P[\X(t_{2x})=(x,x)] &= \P[-\xi_{1x} < \Delta(\x_0,x,x) < \xi_{2x}].
\end{align*}
Repeating the argument in the proof of \prettyref{lemma:difference}, we obtain
\[
\P[\X(t_{2x})=(x,x)]  \sim (r+1)2^{\beta}K(\beta,r,\x_0) t_{2x}^{-\beta}.
\]
Since 
\[
\P[T(\X) = t_{2x}] = \P[\X(t_{2x})=(x,x)]\cdot \P[T_2(\beta,r,x,x)=\infty],
\]
\prettyref{lemma:escape} then yields
\[
\P[T(\beta,r,\x_0) = t_{2x}] \sim  (r-1)2^{\beta}  K(\beta, r,\x_0) t_{2x}^{-\beta}.
\]
Summing over $x$ such that $t_{2x} \geq t$ and using the following Riemann sum approximation,
\[
\sum_{x:t_{2x}\geq t} t_{2x}^{-\beta}\sim \int_t^\infty \frac{1}{2} z^{-\beta} dz = \frac{1}{2(\beta-1)} t^{1-\beta},
\]
we obtain \prettyref{eq:duration-diff-superlinear}.
\myfill \end{proof}

%\subsubsection{Sublinear Case: $0\leq \beta < 1$}
\subsubsection{Proof of \prettyref{thm:duration-sublinear-diff}}
\label{subsubsec:duration-diff-sublinear-proof}

%\begin{proof}[of \prettyref{thm:duration-sublinear-diff}]
Let $\X$ be a $(\beta, r, \x_0)$-urn process and $t_{2x} = 2x - \|\x_0\|_1$ as in the proof of \prettyref{thm:duration-diff-superlinear}. Note that
\[
 \P[\X(t_{2x})=(x,x)]  \leq \P[\Delta(\x_0,x+1,x)>0].
\]
Using the standard argument of exponentiation followed by the application of the Markov inequality as in the proof of Chernoff bound, we obtain, for $s< x_{01}^\beta$,
\[
 \P[\X(t_{2x})=(x,x)]   \leq M(s;\beta,x_{01},x+1) M(-rs;\beta, x_{02}, x),
\]
where
\begin{equation}\label{eq:M}
M(s;\beta,y_1,y_2)  = \prod_{j=y_1}^{y_2-1} \left(1-\frac{s}{j^{\beta}}\right)^{-1}, \quad \text{for } s < y_1^{\beta},
\end{equation}
%\begin{align*}
% \P^\x[D(t_{2m})=0]  & \leq )\E[e^{sS_{1(m+1)}-sS_{2m}}] \\
%& = \prod_{j=x_{01}}^m \E[e^{s\xi_{1j}}] \prod_{j=x_{02}}^{m-1} \E[e^{-s\xi_{2j}}]\\
%& = \prod_{j=x}^m \frac{1}{1-\frac{s}{f_1 j^\beta}} \prod_{j=y}^{m-1} \frac{1}{1+\frac{s}{f_2 j^\beta}},
%\end{align*}
Note that
\[
\log M(s;\beta,x_{01},x+1) = -\sum_{j=x_{01}}^x \log\left(1-\frac{s}{j^\beta}\right)\sim \frac{sx^{1-\beta}}{1-\beta}. 
\]
and
\[
\log M(-rs;\beta,x_{02},x) = -\sum_{j=x_{02}}^{x-1} \log\left(1+\frac{rs}{j^\beta}\right)\sim -\frac{rsx^{1-\beta}}{1-\beta}. 
\]
%which, after taking logarithm, yields
%\[
%-\log\P[E_m]  \geq \sum_{j=x_{01}}^m \log\left(1-\frac{s}{f_1 j^\beta}\right)+ \sum_{j=x_{02}}^{m-1} \log\left(1+\frac{s}{f_2 j^\beta}\right).
%\]
%Using the monotonicity of the functions $z\mapsto \log (1+C z^{-\beta})$, we can lower bound the right-hand side by
%\[
%\int_{x_{01}}^m \log\left(1-\frac{s}{f_1z^\beta}\right) dz+ \int_{x_{02}}^m \log\left(1+\frac{s}{f_2 z^\beta}\right)dz
%\]
%By l'H\^opital's rule,
%\begin{align*}
%& \quad \lim_{m\to\infty} \frac{\int^m \log (1+C z^{-\beta}) dz }{m^{1-\beta}}= \frac{C}{1-\beta}.
%\end{align*}
%Therefore,
%\[
%\liminf_{m\to\infty} -\frac{\log \P[E_m]}{m^{1-\beta}} \geq \frac{1}{1-\beta} \left(\frac{1}{f_2} - \frac{1}{f_1}\right) s.
%\]
Thus
\[
\limsup_{x\to\infty} \frac{\log  \P[\X(t_{2x})=(x,x)]}{x^{1-\beta}}  \leq \frac{(1-r)s}{1-\beta}.
\]
Letting $s\to x_{01}^\beta$, we obtain 
\[
\limsup_{x\to\infty} \frac{\log  \P[\X(t_{2x})=(x,x)]}{x^{1-\beta}}  \leq \frac{1-r}{1-\beta}x_{01}^\beta.
\]
By \prettyref{lemma:escape}, 
\[
\log \P[T_2(\beta,r,x,x)=\infty] \sim \log\frac{r-1}{r+1}= o(x^{1-\beta}).
\] 
Since 
\[
\P[T(\X) = t_{2x}] = \P[\X(t_{2x})=(x,x)] \cdot \P[T_2(\beta,r,x,x)=\infty],
\]
using \prettyref{lemma:escape} and the fact $t_{2x}\sim 2x$, we obtain % \prettyref{eq:duration-sublinear-diff-log-pmf}.
\[
\limsup_{x\to\infty} \frac{\log \P[T(\beta,r,\x_0)=t_{2x}]}{t_{2x}^{1-\beta}} \leq \frac{1-r}{1-\beta} 2^{\beta-1} x_{01}^\beta.
\]
Let $0>C>\frac{1-r}{1-\beta} 2^{\beta-1} x_{01}^\beta$. For all large enough $x$,
\[
\P[T(\X)=t_{2x}] \leq e^{C t_{2x}^{1-\beta}}.
\]
Summing over $x$ such that $t_{2x}\geq t$, we obtain
\[
\P[T(\X)\geq t] = \sum_{x:t_{2x} \geq t} \P[T(\X)= t_{2x}] \leq \frac{1}{2} \int_{t-2}^\infty e^{C s^{1-\beta}} ds.
\]
By repeated application of l'H\^opital's rule, 
\begin{align*}
\limsup_{t\to\infty} \frac{\log \P[T(\X)\geq t]}{ t^{1-\beta}} & \leq \lim_{t\to\infty} \frac{\log \int_{t}^\infty e^{C s^{1-\beta}}ds}{ t^{1-\beta}}\\
& = \lim_{t\to\infty} \frac{-t^\beta e^{C t^{1-\beta}} }{(1-\beta) \int_{t}^\infty e^{C s^{1-\beta}}ds}=C.
\end{align*}
Letting $C \to \frac{1-r}{1-\beta} 2^{\beta-1} x_{01}^\beta$ complets the proof. \myfill \qed
%\end{proof}

\begin{remark}
A modification of the above proof shows that color 1 always wins when $\beta<1$. Indeed, the above proof shows that $\sum_{x} \P[\Delta(\x_0,x+1,x)>0] < \infty$. The Borel-Cantelli Lemma then implies that $\Delta(\x_0, x+1,x)\leq 0$ for all large enough $x$ almost surely, from which it follows that $X_1(t)> X_2(t)$ for large enough $t$. 
\end{remark}

%\begin{theorem}
%For $r>1$ and $\beta=1$,
%\[
% \limsup_{t\to\infty} \frac{\log \P[T\geq t]}{\log t} \leq (1-r)x_0 + 1-r^{-1},
%\]
%and
%\[
%\liminf_{t\to\infty} \frac{\log \P[T\geq t]}{\log t} \geq (1-r)x_0.
%\]
%\end{theorem}
%
%\begin{proof}
%This follows immediately from Theorem 2 of \cite{jiang2015competition}.
%\end{proof}

%% file: TEX/intensity.tex
%%!TEX PS-program = pdflatexmk
%%!TEX root = ../main.tex

\section{Tail Distribution of Intensity}\label{sec:intensity}

In this section, we characterize the tail distribution of intensity $N$.  The equal fitness case ($r=1$) is considered in \prettyref{subsec:intensity-equal}, and the case of different fitnesses ($r>1$) is considered  in \prettyref{subsec:intensity-diff}.

\subsection{Equal Fitness}\label{subsec:intensity-equal}

We consider the equal fitness case in this section. The main results are presented in \prettyref{subsec:intensity-equal-result}, and the proofs are given in \prettyref{subsec:intensity-equal-proof}.

\subsubsection{Main Results}
\label{subsec:intensity-equal-result}

Since $T$ is finite if and only if $N$ is finite, it follows from \prettyref{thm:threshold} that $\P[N(\beta,1,\infty) \geq n] = 1$ for all finite $n$, if $\beta \in [0, 1/2]$. Thus, as in \prettyref{subsec:duration-equal}, our focus in the present section is the regime $\beta > 1/2$.

The following theorem bounds the tail distribution of intensity. For the sublinear regime $\beta\in (1/2,1]$, the tail distribution of $N(\beta,1,\x_0)$ is bounded between two power laws with exponents $\beta$ and $\beta-1/2$, respectively. For the superlinear regime $\beta>1$, we only have an upper bound, but simulations suggest that $N(\beta,1,\x_0)$ also has a power-law tail in this regime.

\begin{theorem}\label{thm:intensity-equal}
\begin{enumerate}
\item[(i)]
For $\beta\in (1/2,1]$,
\begin{equation}\label{eq:intensity-equal-sublinear-upper}
\P[N(\beta,1,\x_0)\geq n]  = O(n^{1/2-\beta}),
\end{equation}
and
\begin{equation}\label{eq:intensity-equal-sublinear-lower}
\P[N(\beta,1,\x_0)\geq n]  = \Omega(n^{-\beta}).
\end{equation}

\item[(ii)]
For $\beta\geq 1$,
\begin{equation}\label{eq:intensity-equal-superlinear-upper}
\P[N(\beta,1,\x_0)\geq n]  = O(n^{-\beta}).
\end{equation}
\end{enumerate}
\end{theorem}

\begin{figure}[thb]
\includegraphics[width=\figwidth]{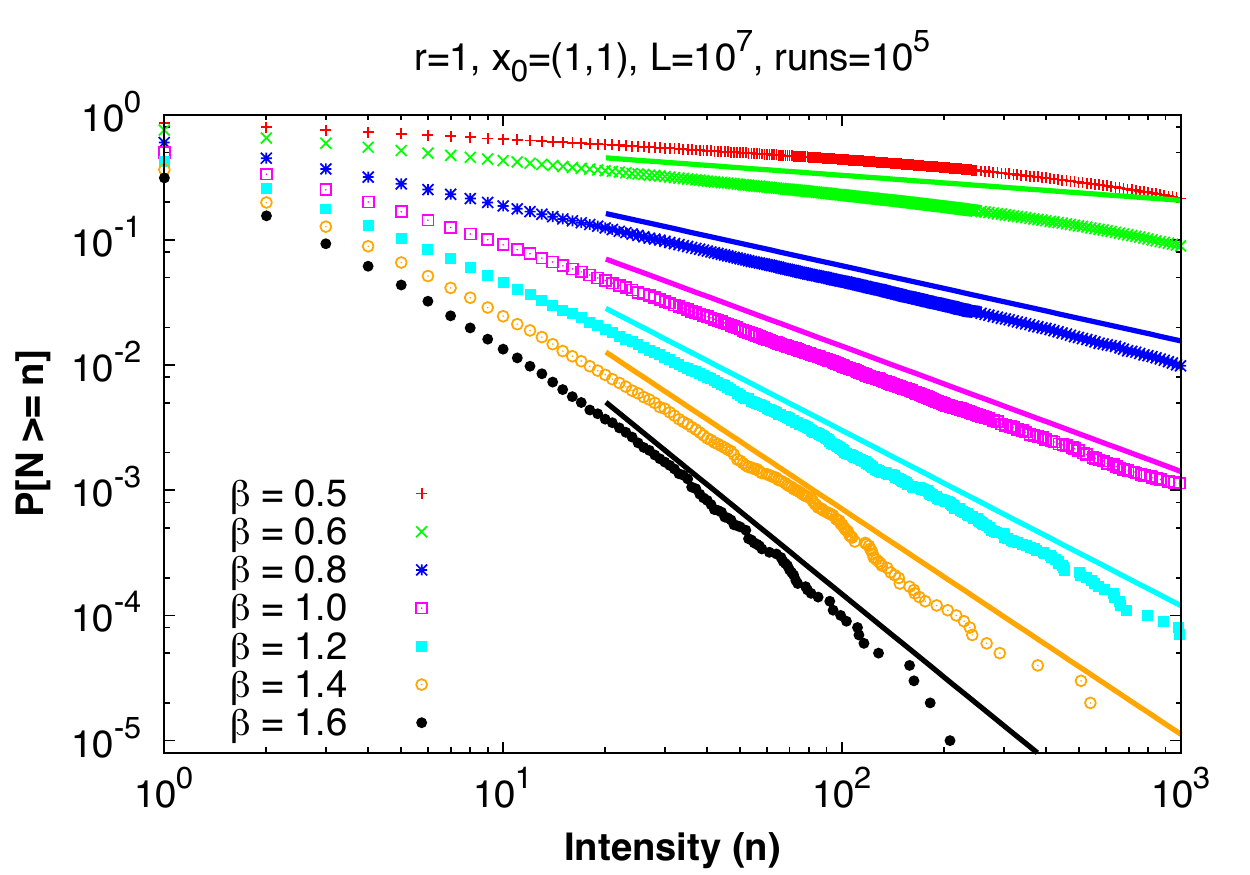}
\caption{
Tail distribution of intensity for $r=1$ and various values of $\beta$. Dots (marks) are simulation results. The solid lines have slopes $1-2\beta$.}
\label{fig:N-equal}
\end{figure}

\prettyref{fig:N-equal} shows the empirical tail distribution of intensity from simulation. Also superimposed are straight lines with slopes $1-2\beta$. Note the good agreement of the simulated slopes with those of the straight lines, which strongly suggests that the intensity $N(\beta,1,\x_0)$ has a power-law tail with exponent $1-2\beta$. We  have the following conjecture,

\begin{conjecture}
For $\beta>1/2$, and some $C(\beta,\x_0)$,
\[
\P[N(\beta,1,\x_0)\geq n]\sim C(\beta,\x_0) n^{1-2\beta}.
\]

\end{conjecture}

\subsubsection{Proof of \prettyref{thm:intensity-equal}}
\label{subsec:intensity-equal-proof}

%\begin{proof}[of \prettyref{thm:intensity-equal}]
By symmetry, we assume $x_{01} \geq x_{02}$ throughout the proof. Note that $N\leq T/2$. Thus for $\beta>1/2$, \prettyref{thm:duration-equal} yields,
\begin{align*}
\P[N(\beta,1,\x_0)\geq n] &\leq \P[T(\beta,1,\x_0)\geq 2n] \sim \frac{K(\beta,1,\x_0)}{\sqrt{(2\beta-1)\pi}}  n^{\frac{1}{2}-\beta},
\end{align*}
which implies \prettyref{eq:intensity-equal-sublinear-upper}.

Now we prove  \prettyref{eq:intensity-equal-sublinear-lower}.
Since $x^\beta$ is concave for $\beta\in (1/2,1]$, by Jensen's inequality, for any $x_1, x_2>0$,
\[
\frac{x_1^\beta + x_2^\beta}{2} \leq \left(\frac{x_1 + x_2}{2}\right)^\beta,
\]
and hence
\[
\frac{x_i^\beta}{x_1^\beta + x_2^\beta} \geq \left(\frac{2x_i}{x_1+x_2}\right)^\beta 2^{-1}.
\]
By considering each sample path in the set $\{T_n(\beta,1,\x_0) = d_0 + 2\ell\}$, where $d_0 = x_{01}-x_{02}$, we obtain
\begin{align}\label{eq:N_sub_lower}
&\quad\ \P[T_n(\beta,1,\x_0) = d_0 + 2\ell]\nonumber \\
& \geq \left[\frac{B(x_{01} + \ell, x_{01} + \ell)}{B(x_{01},x_{02})} 2^{d_0+2\ell} \right]^\beta  \P[T_n(0,1,d_0,0) = d_0 + 2\ell],
\end{align}
where $B(\cdot,\cdot)$ is the beta function.

Note that
\begin{align*}
&\quad\ \P[N(\beta,1,\x_0)\geq n] = \P[T_n(\beta,1,\x_0) <\infty]\\
 &= \sum_{\ell=n-1}^\infty \P[T_n(\beta,1,\x_0)= d_0 +2\ell].
\end{align*}
Then \prettyref{eq:intensity-equal-sublinear-lower}  follows from \prettyref{eq:N_sub_lower} and the following lemma.

\begin{lemma}\label{lemma:N}
Let $\tilde f^{(n,d_0)}_{d_0+2\ell} \triangleq \P[T_n(0,1,d_0,0) = d_0 + 2\ell]$ be the probability that the $n$-th visit to the origin occurs at time $d_0+2\ell$ in a simple random walk starting from $d_0\geq 0$. Then
\begin{equation}\label{eq:N}
\sum_{\ell=n-1}^\infty \left[\frac{B(x_{01} + \ell, x_{01}+ \ell)}{B(x_{01},x_{02})} 2^{d_0+2\ell} \right]^\beta \tilde f^{(n,d_0)}_{d_0+2\ell} =\Theta(n^{-\beta}).
\end{equation}
\end{lemma}

The proof of \prettyref{eq:intensity-equal-superlinear-upper} follows from the same argument as the proof of \prettyref{eq:intensity-equal-sublinear-lower}, except that  the directions of all the inequalities get reversed, since $x^\beta$ is convex for $\beta \geq 1$.

Now we complete the proof of \prettyref{thm:intensity-equal} by proving \prettyref{lemma:N}.

\begin{proof}[of \prettyref{lemma:N}]
Note that for large $\ell$,
\begin{align}\label{eq:beta-asymp}
\left[\frac{B(x_{01} + \ell, x_{01} + \ell)}{B(x_{01},x_{02})} 2^{d_0+2\ell}\right]^\beta &\sim  \frac{C\ \Gamma(2\ell+2x_{01}+1)}{\Gamma(2\ell+2x_{01}+1+\beta/2)},
\end{align}
%\[
%B(x_{01} + \ell, x_{01} + \ell) \sim (2\pi)^{1/2} (2x_{01} + 2\ell)^{-1/2} 2^{-(2x_{01}+2\ell-1)}.
%\]
%Thus
%\begin{align*}
%\left[\frac{B(x_{01} + \ell, x_{01} + \ell)}{B(x_{01},x_{02})} 2^{d_0+2\ell}\right]^\beta & \sim C (2x_{01} + 2\ell)^{-\beta/2}\\
%&\sim  \frac{C\Gamma(2\ell+2x_{01}+1)}{\Gamma(2\ell+2x_{01}+1+\beta/2)},
%\end{align*}
where $\Gamma(\cdot)$ is the gamma function, and
\[
C = \left[\frac{\sqrt{\pi}}{2^{\|\x_0\|_1-3/2} B(x_{01},x_{02})}\right]^\beta.
\]
By Eq.~(4.4.2) of \cite{oldham2006fractional},
\begin{equation}\label{eq:RL-integral}
\frac{\Gamma(2\ell+2x_{01}+1)}{\Gamma(2\ell+2x_{01}+1+\beta/2)}  = {}_0D_1^{-\beta/2} [z^{2x_{01} + 2\ell}],
\end{equation}
where ${}_aD_x^{-\alpha}$ is the Riemann-Liouville fractional integral operator defined by
\[
{}_aD_x^{-\alpha} f = \frac{1}{\Gamma(\alpha)}\int_a^x f(z) (x-z)^{\alpha-1} dz.
\]
Denote the sum in \prettyref{eq:N} by $\Lambda_n$. Combining \prettyref{eq:beta-asymp} and \prettyref{eq:RL-integral} yields
\[
\Lambda_n \sim C \sum_{\ell=n-1}^\infty    {}_0D_1^{-\beta/2} [z^{2x_{01} + 2\ell}]  \tilde  f^{(n,d_0)}_{d_0+2\ell}.
\]
By the linearity of Riemann-Liouville integral for power series (see Section 5.2 of \cite{oldham2006fractional}), 
\begin{align}
\Lambda_n 
&\sim C\ {}_0D_1^{-\beta/2} \left\{ \sum_{\ell=n-1}^\infty    \tilde f^{(n,d_0)}_{d_0+2\ell} z^{2x_{01} + 2\ell}\right\}\nonumber\\
&=C\ {}_0D_1^{-\beta/2} \left[z^{x_{01}+x_{02}} G_n(z;d_0)\right], \label{eq:Lambda_n}
\end{align}
where $G_n(z;d_0) = \sum_{\ell=n-1}^\infty \tilde  f^{(n,d_0)}_{d_0+2\ell}z^{d_0 + 2\ell}$ is the generating function of $\tilde f^{(n,d_0)}_{d_0+2\ell}$, the expression of which is given by the following (see Eq.~(A.15) of \cite{jiang2015competition}),
\begin{equation}\label{eq:mgf}
G_n(z;d_0) = z^{-d_0} \left(1-\sqrt{1-z^2}\right)^{n+d_0-1}.
\end{equation}
Substituting \prettyref{eq:mgf} into \prettyref{eq:Lambda_n} yields
\[
\Lambda_n \sim \frac{C}{\Gamma(\frac{\beta}{2})}\int_0^1 (1-z)^{\frac{\beta}{2}-1}z^{2x_{02}} \left(1-\sqrt{1-z^2}\right)^{n+d_0-1} dz,
\]
where we have used $x_{01}+x_{02}-d_0 = 2x_{02}$.
Note that the integrand can be rewritten as
\[
(1-z^2)^{\frac{\beta}{2}-1} (1+z)^{1-\frac{\beta}{2}} (1+\sqrt{1-z^2})^{x_{02}} (1-\sqrt{1-z^2})^{n+x_{01}-1},
\]
which on $(0,1)$ is bounded between constant multiples of
\[
(1-z^2)^{\frac{\beta}{2}-1}  (1-\sqrt{1-z^2})^{n+x_{01}-1}.
\]
Thus
\[
\Lambda_n =\Theta\left(\int_0^1 (1-z^2)^{\frac{\beta}{2}-1} \left(1-\sqrt{1-z^2}\right)^{n+x_{01}-1} dz\right).
\]
A change of variable $u = \sqrt{1-z^2}$ yields
\begin{align*}
\Lambda_n &=\Theta\left(\int_0^1 u^{\beta-1} (1-u)^{n+x_{01}-1} du\right)\\
&=\Theta(B(\beta, n+x_{01}))=\Theta(n^{-\beta}),
\end{align*}
which completes the proof.
\myfill \end{proof}

\subsection{Different Fitnesses}
\label{subsec:intensity-diff}

We consider the case of different fitnesses in this section. The following theorem shows that the distribution of the intensity $N(\beta,r,\x_0)$ for $r>1$ always has an exponential tail. Thus competitions are never intense when agents have different fitnesses, irrespective of the feedback strength $\beta$ and the initial condition $\x_0$. 

\begin{theorem}\label{thm:intensity-diff}\label{THM:INTENSITY-DIFF}
For $r>1$,
\begin{equation}\label{eq:N-diff}
\P[N(\beta, r, \x_0) \geq n] \leq  r^{-(x_{01}-x_{02})^+}\left(\frac{2}{r+1}\right)^{n-1},
\end{equation}
where $(x)^+=\max\{x,0\}$. In addition,
\begin{equation}\label{eq:N-diff-log}
\lim_{n\to\infty} \frac{\log \P[N(\beta,r,\x_0)\geq n]}{n} = \log\left(\frac{2}{r+1}\right).
\end{equation}
\end{theorem}

\begin{figure}[t]
\centering
\includegraphics[width= \figwidth]{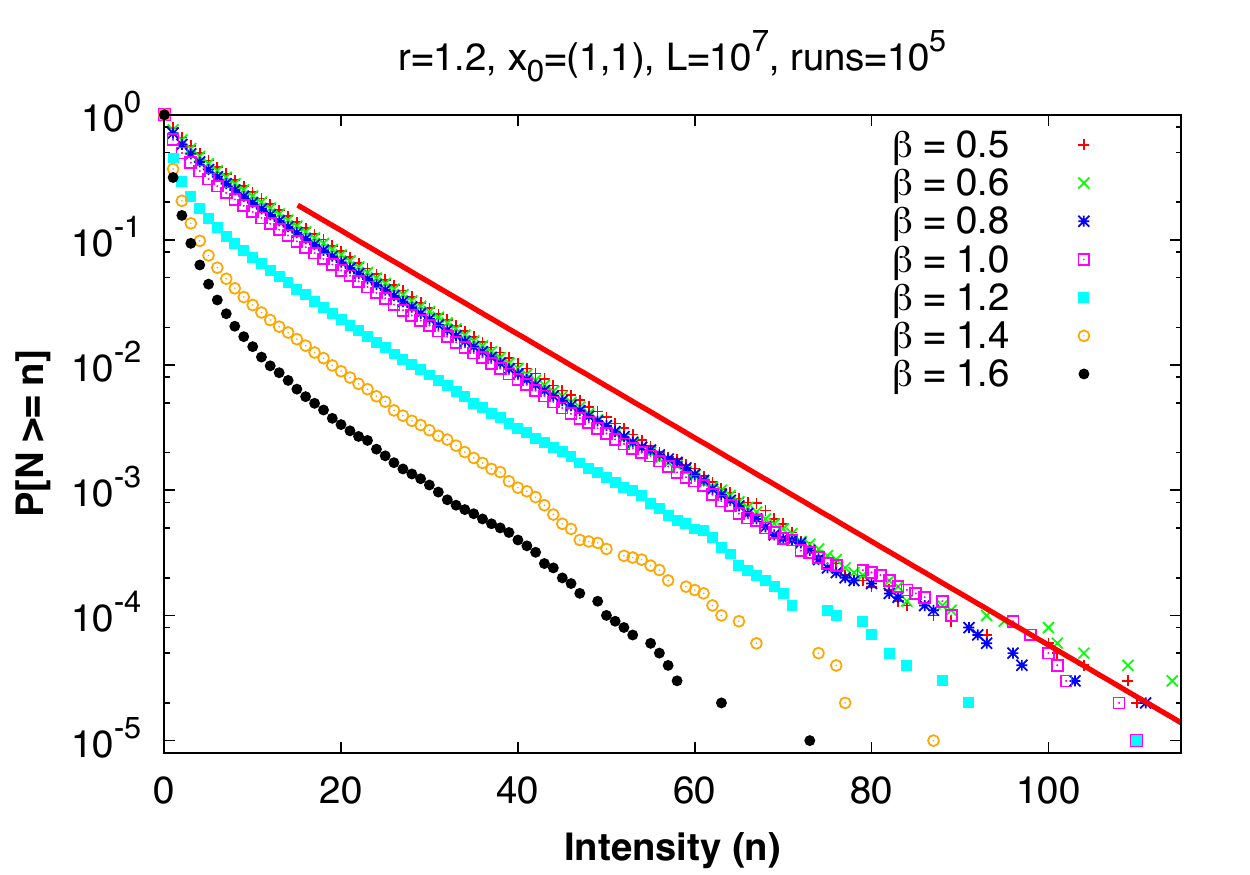}
\caption{
Tail distribution of intensity for $r=1.2$ and various values of $\beta$. The dots (marks) are from simulation. The straight line has slope $\log\frac{2}{r+1}$.}
\label{fig:N-diff}
\end{figure}

The result is illustrated in \prettyref{fig:N-diff}, which shows the simulated tail distributions of intensity. Note that the plot uses semi-log scale. The superimposed straight line has the slope $\log\frac{2}{r+1}$ given in \prettyref{eq:N-diff-log}. Note that the simulated curves all become parallel to the straight line, in good agreement with the theory. Of course, specific $\beta$ values do affect the leading constants, as reflected by the parallel shifts of the curves. 

\begin{proof}[of \prettyref{thm:intensity-diff}]
Eq.~\prettyref{eq:N-diff} follows from \prettyref{corollary:intensity-diff} and the well-known formula for $\P[N(0,r,\x_0)\geq n]$ (see e.g. XI.3.d of \cite{fellerOne}).

Now we prove \prettyref{eq:N-diff-log}. Let $\X$ be a $(\beta, r, \x_0)$-urn process and $F_n(z) = \P[X_1(T_n(\X))=z\mid T_n(\X)<\infty]$. By the strong Markov property and the fact that $F_n(z)=0$ for $z<n$,
\[
\P[T_{n+1}(\X) < \infty\mid T_{n}(\X)<\infty]=\sum_{z\geq n} F_n(z) \P[T_2(\beta,r,z,z)].
\]
\prettyref{lemma:escape} then implies
\begin{equation}\label{eq:eventual-return}
\lim_{n\to\infty} \P[T_{n+1}(\X)<\infty\mid T_{n}(\X)<\infty] = \frac{2}{r+1}.
\end{equation}
Since
\[
\P[N(\beta,r,\x_0)\geq n] =\prod_{j=0}^{n-1} \P[T_{j+1}(\X)<\infty\mid T_{j}(\X)<\infty],
\]
\prettyref{eq:N-diff-log} follows from \prettyref{eq:eventual-return} and the fact that the Ces\`aro mean of a convergent sequence converges to the limit of the sequence.
\myfill \end{proof}

%% file: TEX/conclusion.tex
%%!TEX PS-program = pdflatexmk
%%!TEX root = ../main.tex

\section{Discussion and conclusion}\label{sec:conclusion}

\begin{figure}[t]
\centering
\includegraphics[width=1\columnwidth]{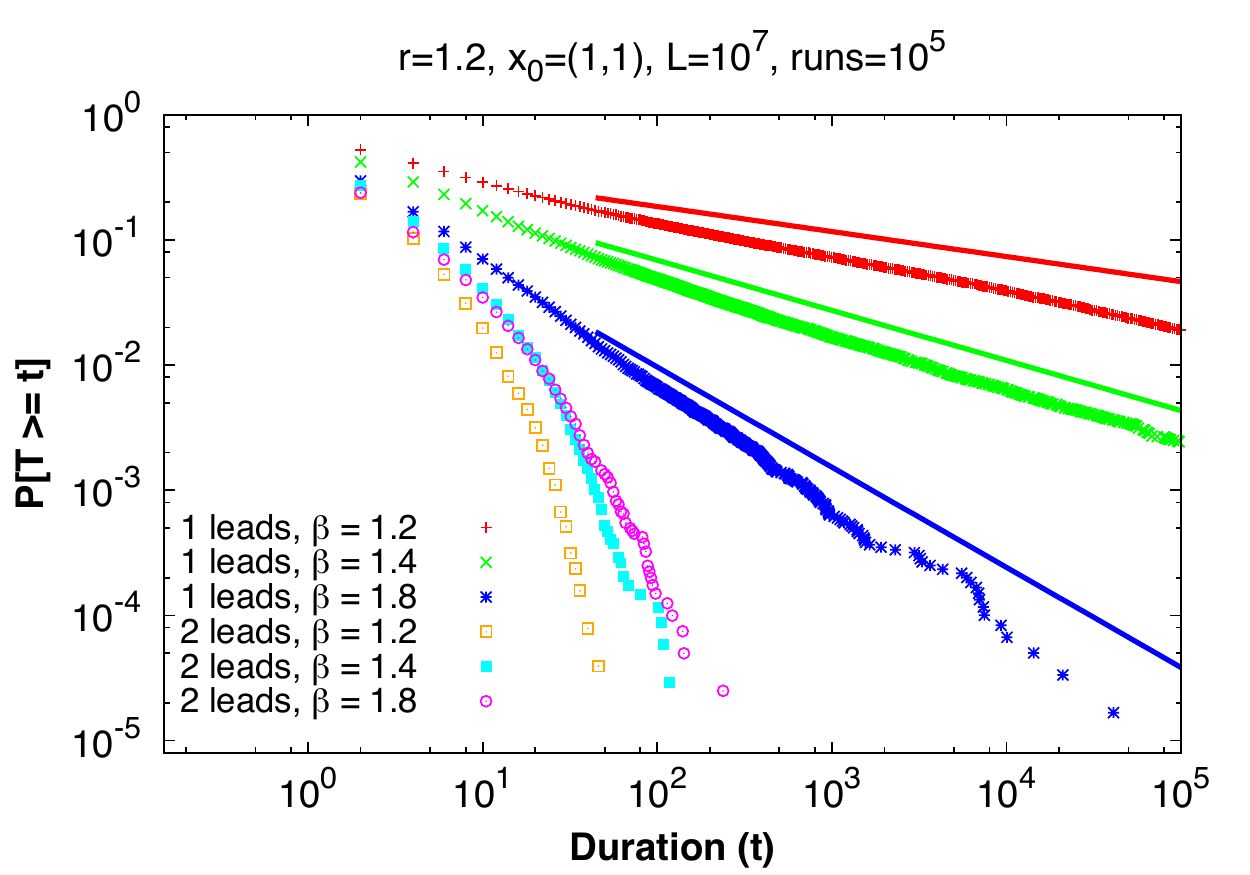}
\caption{Empirical distribution of duration conditioned on either $1$ or $2$ leading the competition at the end of simulation time.}
\label{fig:T_XY_f1d2}
\end{figure}

Apart from the insights provided by the simulations on our theoretical findings, we illustrate another interesting aspect of the different fitness case. Recall that in the superlinear regime the fittest agent can lose the competition. Does the competition duration depend on the winner? \prettyref{fig:T_XY_f1d2} strongly suggests that the answer is yes, which shows the empirical duration distribution conditioned on either $1$ or $2$ leading the competition at the end of the simulation. For competitions that $1$ leads, we observe a power law distribution, consistent with our theoretical findings (the same slopes $1-\beta$ are shown in the plot). However, for competitions that $2$ leads, duration seems to be dominated by an exponential tail. Thus, if $2$ is to win the competition it has to do so early on: $2$ has very little chance of winning if it is trailing behind when a long time has elapsed. However, if $1$ is to win, competitions may last very long with $2$ putting up a good battle for the lead but losing eventually.

This work presented a rigorous mathematical treatment of a nonlinear P{\'o}lya urn process which embodies the fitness of agents ($f_{1,2}$) and the feedback strength of CA effect ($\beta > 0$). In particular, we considered sublinear ($\beta < 1$) and superlinear ($\beta > 1$) regimes as well as equal ($f_{1} = f_{2}$) and non-equal ($f_{1} \neq f_{2}$) fitness scenarios and characterized the tail distribution of two important statistics of competitions: duration (i.e., time of the last tie) and intensity (i.e., number of ties). We characterized the complex interactions between fitness superiority and feedback strength, revealing various interesting properties of such competitions, such as the serious struggle of the fittest in the superlinear regime. 
We believe that our theoretical findings contribute to various applications of the generalized P{\'o}lya urn processes that incorporate both fitness and nonlinearity. %Our results provide a sound theoretical understanding for applying such models in more general contexts. 

%% file: TEX/ack.tex
%%!TEX PS-program = pdflatexmk
%%!TEX root = ../main.tex

\section{Acknowledgements}
This work was supported in part by Army Research Office Contract W911NF-12-1-0385, and ARL Cooperative Agreement W911NF-09-2-0053.
D. Figueiredo received financial support through grants from CAPES, FAPERJ and CNPq (Brazil).

%% file: TEX/appendix.tex
%%!TEX PS-program = pdflatexmk
%%!TEX root = ../main.tex

\section{Proof of \prettyref{thm:dominance-diff}}
\label{app:dominance-diff}

The proof uses a coupling argument similar to the one used in the proof of \prettyref{thm:dominance-equal}. Let $\{\eta_j\}_{j\in \N}$ be a sequence of independent random variables uniformly distributed on $[0,1]$. Define a $(\beta,r,\x_0)$-urn process $\{ \Y(t)\}_{t\in \N}$ recursively by setting $\Y(0) = \x_0$ and
\begin{align*}
 Y_1(t+1) &= Y_1(t) + \ind{\eta_t\leq \frac{r  Y_1(t)^\beta}{r  Y_1(t)^\beta +  Y_2(t)^\beta} },\\
 Y_2(t+1) &=  Y_1(t)+Y_2(t) + 1 - Y_1(t+1).% + \ind{\eta_t > \frac{r  Y_1(t)^\beta}{r  Y_1(t)^\beta +  Y_2(t)^\beta} }.
\end{align*}
Similarly define a $(\beta',r',\x_0')$-urn process $\{ \Y'(t)\}_{t\in\N}$ using the same sequence $\{\eta_j\}_{j\in \N}$.  We now show that $T_1( \Y)\geq T_1( \Y')$ if either $(i)$ or $(ii)$ holds.

 If $x_{01}=x_{02}$, this is trivial. Now assume $(i)$ holds with $x_{01}>x_{02}$.
 We will show by induction that $ Y_1(t) \geq  Y'_1(t)$, $Y'_2(t) \geq  Y_2(t)$ and $Y_1(t)\geq Y_2(t)$ for $t \leq T_1(\Y)$. The base case $t=0$ holds trivially. Assume it holds for $t<T_1(\Y)$ and consider $t+1$. Since $x_{01} > x_{02}$, by the definition of $T_1(\Y)$,  we have $ Y_1(t) \geq  Y_2(t)$ for $t <  T_1(\Y)$. The induction hypothesis then implies that
 \[
 \frac{r  Y_1(t)^\beta Y'_2(t)^{\beta'}}{r'  Y_2(t)^\beta Y'_1(t)^{\beta'}} = \frac{r}{r'} \left(\frac{Y_1(t)}{Y_2(t)}\right)^{\beta-\beta'} \left(\frac{Y_1(t) }{Y'_1(t) }\cdot \frac{ Y'_2(t)}{ Y_2(t)}\right)^{\beta'}\geq 1.
 \]
It follows that
\[
\frac{r  Y_1(t)^\beta}{r  Y_1(t)^\beta +  Y_2(t)^\beta} \geq \frac{r'  Y'_1(t)^{\beta'}}{r'  Y'_1(t)^{\beta'} +  Y'_2(t)^{\beta'}} ,
\]
and hence
\begin{align*}
 &\quad\ Y_1(t+1) = Y_1(t) + \ind{\eta_t\leq \frac{r  Y_1(t)^\beta}{r  Y_1(t)^\beta +  Y_2(t)^\beta} }\\
&\geq   Y'_1(t) + \ind{\eta_t\leq \frac{r'  Y'_1(t)^{\beta'}}{r'  Y'_1(t)^{\beta'} +  Y'_2(t)^{\beta'}}} =  Y'_1(t+1).
\end{align*}
Similarly, $Y_2(t+1)\leq Y'_2(t+1)$, which completes the induction. In particular,
\[
Y_1(t) -  Y_2(t) \geq Y'_1(t) -  Y'_2(t)
\]
for $t \leq T_1(\Y)$. Since $Y_1(0)-Y_2(0)\geq Y'_1(0) -  Y'_2(0) > 0$, it follows that $T_1(\Y) \geq T_1(\Y')$.

Now assume $(ii)$ holds. The same argument as above shows that $ Y_1(t) \leq  Y'_1(t)$,  $Y'_2(t) \leq  Y_2(t)$ and $Y_1(t)\leq Y_2(t)$ for $t\leq T_1(\Y)$, which implies $T_1(\Y) \geq T_1(\Y')$.

\section{Proof of \prettyref{lemma:cf}}
\label{app:cf}

First consider the case $\beta>1$. In this case, it is known (see e.g.~Section 3.2 of \cite{oliveira2008balls}) that  $S_k(x_{0k},\infty)<\infty$ almost surely. The characteristic function of $S_k(x_{0k},\infty)$ is given by $\Psi(s f_k^{-1};\beta,x_{0k},\infty)$, which is absolutely integrable. Thus $S_k(x_{0k},\infty)$ has an absolutely continuous distribution $H_k$ with continuous density $h_k$. Let $f_1=1$ and $f_2=r^{-1}$. Note that $K(\beta,r,\x_0)$ is the probability density of $S_1(x_{01},\infty) - S_2(x_{02},\infty)$ at the origin. By the Convolution Theorem,
\[
K(\beta,r,\x_0) = \int_0^\infty h_1(z) h_2(z) dz\in\R.
\]
Since $h_1$ is not identically zero, $h_1(z_0) > 0$ for some $z_0>0$. By continuity, there exists some $\epsilon>0$ such that $h_1(z)> h_1(z_0)/2$ for $z\in (z_0-\epsilon, z_0+\epsilon)\subset (0,\infty)$.  Thus 
\[
K(\beta,r,\x_0) \geq \frac{h_1(z_0)}{2} \int_{z_0-\epsilon}^{z_0+\epsilon} h_2(z) dz > 0,
\]
where the last inequality holds because every $z\in (0,\infty)$ is a point of increase of $H_2$ by Theorem 3.7.5 of \cite{lukacs1970characteristics}.

Now consider the case $\beta>1/2$ and $r=1$. The proof is similar to that of Theorem 4 in \cite{oliveira2009onset}. By symmetry, assume $x_{01}\leq x_{02}$ without loss of generality. In this case, 
\[
\tilde \Psi(s;\beta,1,\x_0) = \Psi(s;\beta,x_{01},x_{02}) \hat H_3(s;\beta,x_{02}),
\]
where
\[
\hat H_3(s;\beta,x_{02})=\lim_{x\to\infty} |\Psi(s;\beta,x_{02},x)|^2 = \prod_{j=x_{02}}^\infty \left(1+\frac{s^2}{j^{2\beta}}\right)^{-1},
\]
the characteristic function of $\sum_{j=x_{02}}^\infty (\xi_{1j} - \xi_{2j})$, which is finite almost surely (see e.g.~Section 3.2 of \cite{oliveira2008balls}).

If $x_{01} = x_{02}$, then $\Psi(s;\beta,x_{01},x_{02}) = 1$, and $K(\beta,1,\x_0) > 0$ follows from the fact $\hat H_3(s;\beta,x_{02})>0$. 

Suppose $x_{01} < x_{02}$.  Since $\hat H_3$ is absolutely integrable, the corresponding distribution $H_3$ is absolutely continuous with continuous density $h_3$.  Let $H_4$ and $h_4$ be the distribution function and density of $S_1(x_{01}, x_{02})$, both continuous on $(0,\infty)$. By the Convolution Theorem,
\[
K(\beta,1,\x_0) = \int_0^\infty h_3(-z) h_4(z) dz \in \R.
\]
Again by Theorem 3.7.5 of \cite{lukacs1970characteristics}, every $z\in \R$ is a point of increase of $H_3$. Since $h_4$ is continuous and not identically zero, the same argument as for the $\beta>1$ case shows that the above integral is strictly positive.

\section{Proof of \prettyref{LEMMA:FIRST-VISIT}}
\label{app:first-visit}

We will need the next two lemmas that give some large deviation results. Their proofs are deferred to \prettyref{app:Chernoff} and \prettyref{app:large-deviation}, respectively.

\begin{lemma} \label{LEMMA:CHERNOFF}\label{lemma:Chernoff}
Let  $y_m\sim z_m \sim m$, and $q_m =y_m - z_m \geq 9$. For $\epsilon\in(0,1)$ and large enough $m$,
\begin{equation}\label{eq:Chernoff-lower}
\P\left\{S_k(z_m,y_m) < (1-\epsilon) q_m m^{-\beta}\right\} 
\leq e^{-\frac{1}{2}\epsilon\sqrt{q_m}+1},
\end{equation}
and
\begin{equation}\label{eq:Chernoff-upper}
\P\left\{S_k(z_m,y_m) > (1+\epsilon) q_m m^{-\beta}\right\} 
\leq e^{-\frac{1}{2}\epsilon\sqrt{q_m}+\frac{9}{2}}.
\end{equation}
%\begin{enumerate}
%\item[$(i)$]
%For  and large enough $m$,
%\begin{equation}\label{eq:Chernoff-lower}
%\hspace{-.2mm}\P\left\{S_k(z_m,y_m) < (1-\epsilon) q_m m^{-\beta}\right\} 
%\leq e^{-\frac{1}{2}\epsilon\sqrt{q_m}+1},
%\end{equation}
%\item[$(ii)$] For $\epsilon>0$, $q_m=\omega(1)$ and large enough $m$,
%\begin{equation}\label{eq:Chernoff-upper}
%\hspace{-.3mm}\P\left\{S_k(z_m,y_m) > (1+\epsilon) q_m m^{-\beta}\right\} 
%\leq e^{-\frac{1}{2}\epsilon\sqrt{q_m}+\frac{9}{2}}.
%\end{equation}
%\end{enumerate}
\end{lemma}

\begin{lemma} \label{LEMMA:LARGE-DEVIATION}\label{lemma:large-deviation}
For $\beta>0, c>0$, $z_m\geq m+1$ and $q_m\geq 1$,
\begin{align}\label{eq:large-deviation}
\P\left\{\sup_{y\geq z_m} \xi_{ky}  > c q_m m^{-\beta} \right\} = O(m e^{-c q_m}).
\end{align}
\end{lemma}

%\begin{proof}
Now we prove \prettyref{lemma:first-visit}. Let $E\triangleq \{T_1(\beta,1,\x)<\infty\}$. By symmetry, assume $q\triangleq x_1- x_2>0$. The event $E$ occurs if and only if $S_2(x_2,y) < S_1(x_1, y+1)$ for some $y\geq x_1$, i.e.
\begin{align*}
E &=\left\{\sup_{y \geq x_1}\Delta(\x, y+1,y)>0 \right\}\\
&= \left\{\sup_{y\geq x_1}\left[\xi_{1y} + \Delta(x_1,x_1,y,y)\right] > S_2(x_2,x_1)\right\}.
\end{align*}
Let $m = \|\x\|_1/2$. Note that $E \subset E_1 \cup E_2 \cup E_3$, where
\begin{align*}
E_1 &= \left\{\sup_{y\geq x_1} \Delta(x_1,x_1,y,y)  > \left(1-\frac{2}{3}\epsilon\right) q m^{-\beta}\right\},\\
E_2 &= \left\{\sup_{y\geq x_1} \xi_{1y}  >\frac{1}{3}\epsilon  q m^{-\beta} \right\},\\
E_3 & =  \left\{S_2(x_2,x_1)< \left(1-\frac{1}{3}\epsilon\right) q m^{-\beta}\right\}.
\end{align*}
Note that $q=\rho(\x)\sqrt{\|\x\|_1} = \Omega(\sqrt{m})$.  
 By \prettyref{eq:normal-approx-upper}, \prettyref{eq:Chernoff-lower} and \prettyref{eq:large-deviation}, we obtain
\[
\P[E] \leq \P[E_1]+\P[E_2] + \P[E_3] \leq 2\bar\Phi(c_1 \rho(\x)) + O(\|\x\|_1^{-\beta}).
\]
On the other hand, $E_4\subset E_5 \cup E$, where
\begin{align*}
E_4 &= \left\{\sup_{y\geq x_1} \Delta(x_1,x_1,y,y) > (1+\epsilon) q m^{-\beta}\right\},\\
E_5 & =  \left\{S_2(x_2,x_1) > (1+\epsilon) q m^{-\beta}\right\}.
\end{align*}
By \prettyref{eq:normal-approx-lower} and \prettyref{eq:Chernoff-upper}, we obtain 
\[
\P[E] \geq \P[E_4] - \P[E_5] \geq 2\bar\Phi(c_2 \rho(\x)) - O(\|\x\|_1^{-\beta}).
\]
%\end{proof}

\subsection{Proof of \prettyref{lemma:Chernoff}}
\label{app:Chernoff}
We first prove \prettyref{eq:Chernoff-lower}.  Using the standard argument of exponentiation followed by the application of the Markov inequality as in the proof of Chernoff bound, we obtain for $s > 0$,
\begin{align*}
P_1 &\triangleq\P\left\{S_k(z_m,y_m)  < (1-\epsilon) q_m m^{-\beta}\right\}\\
%&=\P\left\{s m^\beta S_k(z_m,y_m)< s(1-\epsilon) q_m\right\}\\
&\leq e^{s(1-\epsilon) q_m} M(-s m^\beta;\beta, z_m, y_m)\\%\prod_{j=z_m}^{y_m-1} \E\left[e^{-s m^\beta \xi_{kj}}\right]\\
%&=e^{s(1-\epsilon) q_m} \prod_{j=z_m}^{y_m-1} \frac{1}{1+s \left(\frac{m}{j}\right)^\beta}\\
&\leq e^{s(1-\epsilon) q_m} \left[1+s \left(\frac{m}{y_m}\right)^\beta\right]^{-(y_m-z_m)},
\end{align*}
where $M$ is given by \eqref{eq:M}.
Let $\kappa_1 =1-\epsilon/2$. Since $y_m\sim m$, for large enough $m$, we have $\left(\frac{m}{y_m}\right)^{\beta} > \kappa_1$, and hence
\[
P_1 \leq e^{s(1-\epsilon) q_m} (1+s \kappa_1)^{-q_m}.
\]
Applying the following inequality to the last term,
\[
(1+u)^{-1} \leq 1-u+u^2 \leq e^{-u+u^2}, \quad \text{for } u\geq 0,
\]
we obtain
%\[
%(1+s c_0)^{-q_m}\leq e^{-c_0 s q_m + c_0^2 q_m s^2}\leq e^{-c_0 s q_m +  q_m s^2}.
%\]
\[
P_1\leq e^{s(1-\epsilon) q_m} e^{q_m(-s \kappa_1 + s^2 \kappa_1^2)}=e^{-\frac{1}{2} \epsilon s q_m + \kappa_1^2 q_m s^2}.
\]
Since $\kappa_1\in(0,1)$, setting  $s=q_m^{-1/2}$ in the above inequality yields \prettyref{eq:Chernoff-lower}.

Now we prove \prettyref{eq:Chernoff-upper}. For large $m$ and $s\in (0, z_m^\beta/m^\beta)$, the standard argument of exponentiation followed by the application of the Markov inequality yields
\begin{align*}
P_2 &\triangleq \P\left\{S_k(z_m,y_m) > (1+\epsilon) q_m m^{-\beta}\right\}\\
%&=\P\left\{s m^\beta S_k(z_m,y_m) > s(1+\epsilon) q_m\right\}\\
&\leq e^{-s(1+\epsilon) q_m} M(sm^\beta;\beta, z_m,y_m)\\%\prod_{j=z_m}^{y_m-1} \E\left[e^{s m^\beta \xi_{kj}}\right]\\
%&=e^{-s(1+\epsilon) q_m} \prod_{j=z_m}^{y_m-1} \frac{1}{1-s \left(\frac{m}{j}\right)^\beta}\\
&\leq e^{-s(1+\epsilon) q_m} \left[1-s \left(\frac{m}{z_m}\right)^\beta\right]^{-(y_m-z_m)},
\end{align*}
%where second equality has used the moment generating function of exponential random variables.  
Let $\kappa_2 =1+\epsilon/2$. Since $z_m\sim m$, for large enough $m$, we have $\left(\frac{m}{z_m}\right)^{\beta} < \kappa_2$,
\[
P_2 \leq e^{-s(1+\epsilon) q_m} (1-s \kappa_2)^{-q_m}.
\]
Applying the following inequality to the last term,
\[
(1-u)^{-1} \leq 1+ u + 2u^2 \leq e^{u+2u^2}, \quad \text{for } u\in \left[0, \frac{1}{2}\right],
\]
we obtain
\[
P_2\leq e^{-s(1+\epsilon) q_m} e^{q_m (s \kappa_2 + 2s^2 \kappa_2^2)}=e^{-\frac{1}{2} \epsilon s q_m + 2 \kappa_2^2 q_m s^2}.
\]
Since $\kappa_2\in(1,3/2)$, setting $s=q_m^{-1/2}$ yields \prettyref{eq:Chernoff-upper}. Note that the conditions that $s\in (0, z_m^\beta/m^\beta)$ and $s \kappa_2\in [0,1/2]$ are satisfied by this particular choice of $s$ when $m$ is large enough.

\subsection{Proof of \prettyref{LEMMA:LARGE-DEVIATION}}
\label{app:large-deviation}

By the union bound, 
\begin{align*}
P&\triangleq \P\left\{\sup_{y\geq z_m} \xi_{ky}  > c  q_m m^{-\beta} \right\} \\
& \leq \sum_{y\geq z_m} \P[\xi_{ky} > c  q_m m^{-\beta}]\\
&= \sum_{y\geq z_m} e^{-c  q_m m^{-\beta}  y^{\beta} }.
\end{align*}
Since the summand is decreasing in $y$ and $z_m \geq m+1$, bounding the sum by the corresponding integral yields
\begin{align*}
P&\leq \int_{m}^\infty e^{-c q_m m^{-\beta} z^\beta} dz\\
&= \int_{1}^\infty m e^{-c q_m z^\beta} dz \\
&= m e^{-c q_m} \int_1^\infty e^{-cq_m  (z^\beta-1)} dz\\
&\leq m e^{-c q_m} \int_1^\infty e^{-c (z^\beta-1)} dz.
\end{align*}
Since the last integral is finite, $P= O(m e^{-c q_m})$.

\section{Proof of Uniform Convergence in \prettyref{LEMMA:DIFFERENCE}}
\label{app:uniform}

Throughout this section, the limiting process is understood to be $t\to\infty$. For a function $G(\x,\dots)$ of $\x$ and other variables, we will use the following notation,
\[
\|G(\x,\dots)\|_\x \triangleq \sup_{\x\in A_t(t^\gamma)} |G(\x,\dots)|.
\]

Recall that we have shown in \prettyref{subsec:duration-equal-proof}
that
\[
 \P[\X(t)=\x] = \frac{x_1^{-\beta} + x_2^{-\beta}}{2\pi}\int_{-\infty}^\infty  \psi(s;x_1+1,x_2+1) ds,
\]
which can be rewritten as
\[
\P[\X(t)=\x] = \frac{2^{\beta}}{\pi t^{\beta}}\int_{-\infty}^\infty \tilde \Psi(s;\beta,r,\x_0) Z(\x,t,s)ds,
\]
where $\tilde \Psi(s;\beta,r,\x_0)$ is defined in \eqref{eq:Psi-tilde}, and
\[
Z(\x,t,s) =\frac{1}{2}\left[\left(\frac{t}{2x_1}\right)^\beta + \left(\frac{t}{2x_2}\right)^\beta \right]
 \frac{\psi(s;x_1+1,x_2+1)}{\tilde \Psi(s;\beta,r,\x_0)}.
\]
Recalling the definition \eqref{eq:K} of $K(\beta,r,\x_0)$, we obtain
\begin{align}
&\quad\; \left\| t^\beta \P[\X(t)=\x] - 2^{\beta+1} K(\beta,r,\x_0)\right\|_\x \nonumber\\
&=\frac{2^{\beta}}{\pi} \left\| \int_{-\infty}^\infty \tilde \Psi(s;\beta,r,\x_0)[Z(\x,t,s) - 1]ds \right\|_\x \nonumber\\
&\leq \frac{2^{\beta}}{\pi} \int_{-\infty}^\infty  |\tilde \Psi(s;\beta,r,\x_0)|\cdot \|Z(\x,t,s)-1\|_\x\, ds.\label{eq:bound}
\end{align}
Since for all large $t$, the last integrand is upper bounded by
\[
2|\psi(s;x_{01}+1,x_{02}+1)|\leq 2\left(1+\frac{s^2}{\|x_0\|_1^{2\beta}}\right)^{-1},
\]
if we can show
\begin{equation}\label{eq:uniform}
\|Z(\x,t,s) - 1\|_\x \to 0, 
\end{equation}
then the uniform convergence claimed in \prettyref{lemma:difference} will follow from \eqref{eq:bound} and the Dominated Convergence Theorem.

Now we prove \eqref{eq:uniform}. Rewrite $Z$ in polar form as $Z(\x,t,s) = R(\x,t,s) e^{i\Theta(\x,t,s)}$, i.e.~$R(\x,t,s) = |Z(\x,t,s)|$ and $\Theta(\x,t,s) = \arg Z(\x,t,s)$. Note that
\[
 \frac{\psi(s;x_1+1,x_2+1)}{\tilde \Psi(s;\beta,r,\x_0)} = \prod_{j=x_1+1}^\infty \left(1-\frac{is}{j^{\beta}}\right)  \prod_{j=x_2+1}^\infty \left(1+\frac{is}{j^{\beta}}\right).
\]
Thus
\[
R(\x,t,s) \geq \frac{1}{2}\left[\left(\frac{t}{2x_1}\right)^\beta + \left(\frac{t}{2x_2}\right)^\beta \right],
\]
and
\[
R(\x,t,s) \leq \frac{1}{2}\left[\left(\frac{t}{2x_1}\right)^\beta + \left(\frac{t}{2x_2}\right)^\beta \right] \prod_{j=x_1\wedge x_2}^{\infty} \left(1+\frac{s^2}{j^{2\beta}}\right).
\]
For $\x\in A_t(t^\gamma)$ and large enough $t$, we have $|2x_{1,2}-t|\leq 2 t^{\gamma+1/2} < t$ and $x_1\wedge x_2 \geq \lfloor t/4\rfloor\geq 1$. It follows that
\begin{equation}\label{eq:R-lower}
R(\x,t,s) \geq \left(1+2 t^{\gamma-1/2}\right)^{-\beta},
\end{equation}
and
\begin{equation}\label{eq:R-upper}
R(\x,t,s) \leq \left(1-2 t^{\gamma-1/2}\right)^{-\beta}\prod_{j=\lfloor t/4\rfloor}^{\infty} \left(1+\frac{s^2}{j^{2\beta}}\right).
\end{equation}
Since for $\beta>1/2$,
\[
\prod_{j=\lfloor t/4\rfloor}^{\infty} \left(1+\frac{s^2}{j^{2\beta}}\right) \leq \exp\left(s^2 \sum_{j=\lfloor t/4\rfloor}^{\infty} j^{-2\beta}\right)\to 1,
\]
\eqref{eq:R-lower} and \eqref{eq:R-upper} imply that
\[
\|R(\x,t,s)-1\|_\x \to 0. 
\]

For the phase $\Theta(\x,t,s)$, note that
\[
\Theta(\x,t,s) = (-1)^{\ind{x_1>x_2}} \sum_{j=x_1\wedge x_2+1}^{x_1\vee x_2} \arctan\left(\frac{s}{j^{\beta}}\right),
\]
where $x_1\vee x_2 = \max\{x_1,x_2\}$.
For $\x\in A_t(t^\gamma)$ and large $t$, $|x_1-x_2| \leq 2 t^{\gamma+1/2}$ and $x_1\wedge x_2\geq t/4$. Thus
\begin{align*}
|\Theta(\x,t,s)| &\leq \sum_{j=x_1\wedge x_2+1}^{x_1\vee x_2} \frac{s}{j^{\beta}}\leq s 2^{2\beta+1} t^{\gamma+1/2-\beta}.
\end{align*}
Since $\gamma<\beta-1/2$, it follows that 
\[
\|\Theta(\x,t,s)\|_\x \to 0.
\]

Note that for $z\in \mathbb{C}$,
\begin{align*}
|z - 1|^2 & = (|z|-1)^2 + 4|z| \sin^2\left(\frac{\arg z}{2}\right)\\
& \leq (|z|-1)^2 + |z|\cdot |\arg z|^2.
\end{align*}
It follows that
\begin{align*}
\|Z(\x,t,s)-1\|_\x^2 & \leq \|R(\x,t,s)-1\|_\x^2 \\
&\quad + \|R(\x,t,s)\|_\x \cdot \|\Theta(\x,t,s)\|_\x^2\to 0,
\end{align*}
which completes the proof.

%% file: main.bbl
\begin{thebibliography}{10}

\bibitem{Arthur:94}
W.~B. Arthur.
\newblock {\em Increasing Returns and Path Dependence in the Economy}.
\newblock U.\ Michigan Press, 1994.

\bibitem{Barab99}
A.-L. Barab{\'a}si and R.~Albert.
\newblock Emergence of scaling in random networks.
\newblock {\em Science}, 286(5439):509--512, 1999.

\bibitem{cattuto2007semiotic}
C.~Cattuto, V.~Loreto, and L.~Pietronero.
\newblock Semiotic dynamics and collaborative tagging.
\newblock {\em Proceedings of National Academy of Sciences}, 104(5):1461--1464,
  2007.

\bibitem{davis1990reinforced}
B.~Davis.
\newblock Reinforced random walk.
\newblock {\em Probability Theory and Related Fields}, 84(2):203--229, 1990.

\bibitem{SollaPrice76}
D.~de~Solla~Price.
\newblock A general theory of bibliometric and other cumulative advantage
  processes.
\newblock {\em Journal of the American Society for Information Science},
  27(5):292--306, 1976.

\bibitem{diprete2006cumulative}
T.~A. DiPrete and G.~M. Eirich.
\newblock Cumulative advantage as a mechanism for inequality: A review of
  theoretical and empirical developments.
\newblock {\em Annual review of sociology}, pages 271--297, 2006.

\bibitem{Drinea:2002}
E.~Drinea, A.~Frieze, and M.~Mitzenmacher.
\newblock Balls and bins models with feedback.
\newblock In {\em ACM-SIAM Symposium on Discrete Algorithms (SODA)}, pages
  308--315, 2002.

\bibitem{Polya23}
F.~Eggenberger and G.~P{\'o}lya.
\newblock {\"U}ber die statistik verketteter vorg{\"a}nge.
\newblock {\em Journal of Applied Mathematics and Mechanics/Zeitschrift f{\"u}r
  Angewandte Mathematik und Mechanik}, 3(4):279--289, 1923.

\bibitem{fellerOne}
W.~Feller.
\newblock {\em An introduction to probability theory and its applications},
  volume~1.
\newblock John Wiley \& Sons, 3rd edition, 1968.

\bibitem{golder2006usage}
S.~A. Golder and B.~A. Huberman.
\newblock Usage patterns of collaborative tagging systems.
\newblock {\em Journal of information science}, 32(2):198--208, 2006.

\bibitem{gupta2010survey}
M.~Gupta, R.~Li, Z.~Yin, and J.~Han.
\newblock Survey on social tagging techniques.
\newblock {\em ACM SIGKDD Explorations Newsletter}, 12(1):58--72, 2010.

\bibitem{halpin2007complex}
H.~Halpin, V.~Robu, and H.~Shepherd.
\newblock The complex dynamics of collaborative tagging.
\newblock In {\em Proceedings of 16th international conference on World Wide
  Web}, pages 211--220. ACM, 2007.

\bibitem{jiang2015competition}
B.~Jiang, L.~Sun, D.~Figueiredo, B.~Ribeiro, and D.~Towsley.
\newblock On the duration and intensity of cumulative advantage competitions.
\newblock {\em Journal of Statistical Mechanics: Theory and Experiment},
  2015(11):P11022, 2015.

\bibitem{khanin2001probabilistic}
K.~Khanin and R.~Khanin.
\newblock A probabilistic model for the establishment of neuron polarity.
\newblock {\em Journal of Mathematical Biology}, 42(1):26--40, 2001.

\bibitem{lukacs1970characteristics}
E.~Lukacs.
\newblock {\em Characteristics functions}.
\newblock Griffin, 1970.

\bibitem{Mahm08}
H.~Mahmoud.
\newblock {\em P{\'o}lya urn models}.
\newblock CRC Press, 2008.

\bibitem{Merton68}
R.~K. Merton.
\newblock The matthew effect in science.
\newblock {\em Science}, 159:56--63, 1968.

\bibitem{oldham2006fractional}
K.~B. Oldham and J.~Spanier.
\newblock {\em The fractional calculus: integrations and differentiations of
  arbitrary order}.
\newblock Dover Publications, 2006.

\bibitem{oliveira2008balls}
R.~I. Oliveira.
\newblock Balls-in-bins processes with feedback and {B}rownian {M}otion.
\newblock {\em Combinatorics, Probability and Computing}, 17:87--110, 1 2008.

\bibitem{oliveira2009onset}
R.~I. Oliveira.
\newblock The onset of dominance in balls-in-bins processes with feedback.
\newblock {\em Random Structures \& Algorithms}, 34(4):454--477, 2009.

\bibitem{Peman07}
R.~Pemantle.
\newblock A survey of random processes with reinforcement.
\newblock {\em Probability Surveys}, 4(1-79):25, 2007.

\bibitem{resnick1992adventures}
S.~I. Resnick.
\newblock {\em Adventures in stochastic processes}.
\newblock Birkh{\"a}user, 1992.

\bibitem{sakhanenko2006estimates}
A.~I. Sakhanenko.
\newblock Estimates in the invariance principle in terms of truncated power
  moments.
\newblock {\em Siberian Mathematical Journal}, 47(6):1113--1127, 2006.

\bibitem{wagner2014semantic}
C.~Wagner, P.~Singer, M.~Strohmaier, and B.~A. Huberman.
\newblock Semantic stability in social tagging streams.
\newblock In {\em Proceedings of 23rd international conference on World Wide
  Web}, pages 735--746. ACM, 2014.

\bibitem{Walls12}
T.~Wallstrom.
\newblock The equalization probability of the {P{}\'o}lya urn.
\newblock {\em The American Mathematical Monthly}, 119(6):516--518, 2012.

\bibitem{zhu2009nonlinear}
T.~Zhu.
\newblock {\em Nonlinear P{\'o}lya urn models and self-organizing processes}.
\newblock PhD thesis, University of Pennsylvania, 2009.

\end{thebibliography}
